\documentclass[aps,prd,twocolumn,superscriptaddress,showpacs,longbibliography]{revtex4-1}

\usepackage{amsmath}
\usepackage{amssymb}
\usepackage{graphicx}
\usepackage{hyperref}
\usepackage{appendix}
\usepackage{bm}
\usepackage{amsthm}
\usepackage{xcolor}
\usepackage{comment}

\newtheorem{theorem}{Theorem}[section]

\newcommand{\Tr}{\operatorname{Tr}}

\newcommand{\Pij}{\Psi_{ij}}
\newcommand{\vq}{\bm{q}}

\newcommand{\calP}{\mathcal{P}}
\newcommand{\calN}{\mathcal{N}}
\newcommand{\half}{\tfrac{1}{2}}
\newcommand{\lth}{\lambda_{\mathrm{th}}}

\newcommand{\dif}{\mathrm{d}}

\newcommand{\sig}{\sigma}

\newcommand{\FC}{F_{\rm C}}
\newcommand{\FV}{F_{\rm V}}
\newcommand{\FF}{F_{\rm F}}
\newcommand{\FW}{F_{\rm W}}
\newcommand{\FVC}{F_{\rm C}^{(V)}}
\newcommand{\FVV}{F_{\rm V}^{(V)}}
\newcommand{\FVF}{F_{\rm F}^{(V)}}
\newcommand{\FVW}{F_{\rm W}^{(V)}}
\newcommand{\fMC}{f_{\rm C}^{(M)}}
\newcommand{\fMV}{f_{\rm V}^{(M)}}
\newcommand{\fMF}{f_{\rm F}^{(M)}}
\newcommand{\fMW}{f_{\rm W}^{(M)}}
\newcommand{\Rcv}{\mathcal{R}_{\rm cv}}
\newcommand{\Rfw}{\mathcal{R}_{\rm fw}}
\newcommand{\neff}{n_{\rm eff}}
\newcommand{\dH}{\Delta H}
\newcommand{\Sgrav}{S_3^{\rm grav}}
\newcommand{\Spng}{S_3^{\rm PNG}}
\newcommand{\Stot}{S_3^{\rm tot}}
\newcommand{\fnl}{f_{\rm NL}}

\newcommand{\eps}{\varepsilon_3}

\begin{document}

\title{The Information Content of the Cosmic Web}

\author{Juan Garc\'ia-Bellido}
\affiliation{Instituto de F\'isica Te\'orica UAM/CSIC, Universidad Aut\'onoma de Madrid, Cantoblanco 28049 Madrid, Spain}

\date{\today}

\begin{abstract}
We present a unified information-theoretic treatment of the Cosmic Web built on the full structure of the tidal deformation tensor.  The three eigenvalues
$(\lambda_1,\lambda_2,\lambda_3)$ of the tidal Hessian furnish a natural morphological classifier -- clusters, filaments, walls and voids correspond
to $(+,+,+)$, $(+,+,-)$, $(+,-,-)$ and $(-,-,-)$ sign patterns -- and their joint probability distribution, known analytically in the linear regime from Doroshkevich (1970), defines a continuous Shannon entropy that quantifies the information encoded in the geometry of large-scale structure. First, the Doroshkevich distribution possesses an exact $\mathbb{Z}_2$ symmetry under $\delta\to-\delta$, which equates the cluster and void volume fractions ($F_{\rm C}^{(V)}=F_{\rm V}^{(V)}=7.96\%$) and the filament and wall fractions ($F_{\rm F}^{(V)}=F_{\rm W}^{(V)}=42.04\%$) at every smoothing scale and redshift for which the field remains Gaussian, fixing the Gaussian classification entropy at $H^{(V)}_G=1.632$ bits.  Second, additional information resides in the shear invariants $\mathcal{Q}=\Tr(\mathbf{T}^2)$ and $\mathcal{A}=\Tr(\mathbf{T}^3)$, where $\mathbf{T}\equiv\bm\Psi-\tfrac13(\Tr\bm\Psi)\, \mathbf{1\!\!l}$ is the traceless part of the tidal tensor (the tidal shear).  These are algebraically independent of the density contrast and comprise five of the six independent tensor components, so a \emph{local} (one-point) density summary retains only one of them, a factor of five in differential entropy at equal variance, though the actual ratio depends on the smoothing scale.  We derive the \emph{Log-Doroshkevich distribution}, the exact eigenvalue PDF of the lognormal model, and prove that the $\mathbb{Z}_2$ volume-fraction symmetry survives the nonlinear map, broken only in density-weighted (matter) fractions. Gravitational collapse breaks the symmetry through the second-order skewness of the density PDF: an Edgeworth treatment yields the cluster-void asymmetry $\mathcal{R}_{\rm cv}-1=0.609\,\varepsilon_3$ and the Shannon entropy difference $\Delta H_{\rm CV}=0.107\,\varepsilon_3$ bits, with $\varepsilon_3=S_3^{\rm tot}(R)\,\sigma_0(R)\,D(z)$ the dimensionless third-cumulant amplitude.  The total skewness receives a gravitational contribution $S_3^{\rm grav}(R)=34/7+n_{\rm eff}(R)$ and a primordial non-Gaussianity contribution $S_3^{\rm PNG}(R)=S_3^{(1)}(R)\,f_{\rm NL}$ with $S_3^{(1)}(R=8\,h^{-1}{\rm Mpc})\simeq3\times10^{-4}$; the latter is smaller than the gravitational term by four orders of magnitude, so the one-point cluster--void asymmetry is dominated by gravity and provides only a weak, complementary handle on $f_{\rm NL}$ compared with the scale-dependent halo bias.  The same breaking, tracked in time, drives the configuration-space entropy down from its Gaussian value as voids grow to dominate the volume and matter concentrates in filaments and clusters.  The multifractal information dimension and the redshift evolution of the entropy close the picture: the multifractal entropy rate obeys $\dot H_{\rm mf}(z)=-3f(z)/(1+z)$, equating the information-theoretic entropy rate to the linear growth rate and providing a probe complementary to redshift-space-distortion measurements of $f\sigma_8$.
\end{abstract}

\pacs{98.65.-r, 89.70.Cf, 98.80.-k, 05.45.Df}
\keywords{large-scale structure, Cosmic Web, Shannon entropy, tidal tensor,
shear invariants, $\mathbb{Z}_2$ symmetry, non-Gaussianity, fractal dimension, multifractal, information theory}

\preprint{IFT-UAM/CSIC-26-66}

\maketitle

\section{Introduction}
\label{sec:intro}

The large-scale structure of the Universe, the Cosmic Web, organises
itself into a hierarchy of morphological components: vast, nearly empty voids;
thin planar walls; elongated filaments; and dense, compact galaxy clusters at
their intersections~\cite{Bond1996,Springel2006}.  This four-component
classification, pioneered by the Zel'dovich approximation~\cite{Zeldovich1970}
and formalized by the tidal-field eigenvalue signature~\cite{Hahn2007} and the
\textsc{nexus} scheme~\cite{Cautun2013}, provides a natural partition of the
cosmic density field.  Over the past decade, redshift surveys covering ever
larger volumes~\cite{DESI2024PNG,Euclid2024} have confirmed what
simulations~\cite{Springel2005} and analytic theory~\cite{Bond1996,Shandarin1983}
long predicted: matter segregates into a network of dense clusters connected by
filaments, bounded by sheet-like walls, and punctuated by vast voids.  The
statistical characterization of the web matters not only as a description of
large-scale structure but as a precision cosmological tool, since the geometry
of the web encodes information about the growth history, the nature of gravity,
and the initial conditions of the universe.

The density contrast $\delta\equiv\rho/\bar\rho-1$ is the most commonly used
\emph{scalar} characterization of large-scale structure.  {Over the past
fifteen years, however, the field has moved well beyond isotropic density
statistics, adopting morphological and topological web finders that isolate the
anisotropic components directly.} The density contrast $\delta$ captures
only the trace of the tidal deformation tensor, $I_1\equiv\Tr(\bm\Psi)$; the
two remaining independent invariants, the squared shear
$I_2\equiv[(\Tr\bm\Psi)^2-\Tr(\bm\Psi^2)]/2$ and the cubic invariant
$I_3\equiv\Tr(\bm\Psi^3)$, encode anisotropic information about the geometry
of collapse that is invisible to the density alone.  The tidal
classification~\cite{Hahn2007}, which assigns morphological type according to
the signs of the eigenvalues $\lambda_1\ge\lambda_2\ge\lambda_3$ of the tidal
deformation tensor $\Psi_{ij}=\partial^2\phi/\partial q_i\partial q_j$ at
threshold $\lth=0$, is arguably the most direct connection between the
Lagrangian dynamics and the observed web: a cluster has $(+,+,+)$, a filament
$(+,+,-)$, a wall $(+,-,-)$, a void $(-,-,-)$.  Alternatives based on the
velocity-shear tensor~\cite{ForeRomero2009,Hoffman2012}, persistent
topology~\cite{Sousbie2011}, the Hessian of the density
field~\cite{AragonCalvo2010}, Bayesian sampling~\cite{Tempel2015}, or the
gravitational mapping between the cosmic web and the dark-matter halo
distribution~\cite{Kitaura:2020lkj} have all
been explored and compared systematically~\cite{Libeskind2018}, among
them the persistent-homology and Morse-theory framework \textsc{DisPerSE}~\cite{Sousbie2011},
the watershed-based \textsc{Spineweb}~\cite{AragonCalvo2010}, and the skeleton
formalism~\cite{Novikov2006,Sousbie2008}; the tidal
classification is distinguished by its direct derivation from the Lagrangian
potential, which gives it a well-defined probability distribution in Gaussian initial conditions, the Doroshkevich~\cite{Doroshkevich1970} distribution.

Several authors have brought information-theoretic tools to bear on large-scale
structure.  Hosoya, Buchert \& Morita~\cite{Hosoya2004} introduced the
Kullback--Leibler relative information entropy as a measure of distinguishability
of the inhomogeneous density field from its spatial average; subsequent work
extended this to nonlinear clustering and the cosmological
constant~\cite{Pandey2015}.  Pandey~\cite{Pandey2013} used Shannon entropy of
the density field to test cosmic homogeneity in the SDSS galaxy catalog.
Leclercq et al.~\cite{Leclercq2016} introduced a decision-theoretic framework
that uses mutual information and the Jensen--Shannon divergence to compare
T-web, DIVA and ORIGAMI classifiers~\cite{Falck:2012ai}, asking which
segmentation retains the most cosmological information.  Graph entropy of the
cosmic-web network discriminates cosmological parameters at the $10^{-2}$-bit
level~\cite{Coutinho2016}.  Vazza~\cite{Vazza2017,Vazza2020} applied statistical
complexity and block entropy to cosmological simulations, estimating that the
thermal and kinetic energy fields require $\sim10^{16}$--$10^{17}$ bits to
describe the intergalactic medium at $\sim100\,\mathrm{kpc}$ resolution.  The
multifractal geometry of galaxy clustering, described by the generalized fractal
dimension $D_q$ and the singularity spectrum $g(\alpha)$, has been measured from
galaxy surveys and $N$-body simulations~\cite{Gaite2019}, and Fisher information
applied to tidal-web power spectra improves neutrino mass constraints by a
factor of order $15$ over density-only analyses~\cite{Bonnaire2022}.

{This paper assembles these threads into a single, internally consistent
picture, organized around two structural facts about the tidal field.

\emph{The $\mathbb{Z}_2$ symmetry of the Gaussian baseline.}  The joint
probability density of the three ordered tidal eigenvalues for a Gaussian random
field, derived by Doroshkevich in 1970, possesses an exact discrete symmetry:
the sign reversal $\lambda_i\to-\lambda_i$ (equivalently $\delta\to-\delta$) is
a measure-preserving involution of the Doroshkevich PDF.  As a consequence the
volume fractions of clusters and voids are exactly equal, as are those of
filaments and walls, for \emph{any} smoothing scale and \emph{any} redshift
so long as the field remains Gaussian and the threshold is zero.  The symmetry
group is $\mathbb{Z}_2$.  It fixes the Gaussian classification entropy at
$H^{(V)}_G=1.632$ bits and defines the maximum-symmetry state from which
gravitational evolution departs.}

\emph{The information beyond density.}  {The traceless shear tensor
$\mathbf{T}$ comprises five of the six independent components of $\bm\Psi$, so a
\emph{local} (one-point) density summary retains only one of them; at equal
variance this amounts to five times the differential entropy of $\delta$ (a
scale-dependent statement, Sec.~\ref{ssec:hshear}).  The full density
\emph{field}, by contrast, fixes $\bm\Psi$ exactly through the Poisson equation,
so this information is discarded by local density statistics, not by field-level
analysis.}  The invariants $\mathcal{Q}=\Tr(\mathbf{T}^2)$ and
$\mathcal{A}=\Tr(\mathbf{T}^3)$ are statistically independent of the density and
resolve morphological information invisible to {local} density statistics, while
the multifractal information dimension $D_1$ connects this entropy to the
fractal geometry of the matter distribution.

{The central physical story is the \emph{breaking} of the $\mathbb{Z}_2$
symmetry.  It is exact only for a zero-mean Gaussian field, and is broken,
progressively as a function of redshift and scale, by physically distinct
mechanisms.  Second-order gravity generates a non-zero
skewness~\cite{Bernardeau1994,Bernardeau2002}: the density field develops a
positive tail as overdense regions collapse while underdense regions evacuate
only to $\delta=-1$.  This asymmetry, captured by the Edgeworth expansion of the
one-point PDF, shifts the cluster and void fractions by equal and opposite
amounts and generates a cluster-void Shannon entropy difference $\dH_{\rm CV}$.
If the initial fluctuations deviate from Gaussianity, parametrized at leading
order by the local amplitude $\fnl$~\cite{Maldacena2003,Verde2000}, an
additional skewness $\Spng(R)=S_3^{(1)}(R)\,\fnl$ appears, with
$S_3^{(1)}(8\,h^{-1}{\rm Mpc})\simeq3\times10^{-4}$~\cite{Verde2000,LoVerde2008,Grossi2009};
it shares the temporal scaling of the gravitational term but, being suppressed by
the primordial amplitude $\sigma_\Phi\sim10^{-5}$, is far too small to compete with
scale-dependent bias as a one-point $\fnl$ estimator.  Tracked in time, the same breaking drives the
configuration-space entropy down from its Gaussian value as voids grow to
dominate the volume, with a rate set by the linear growth rate $f(z)$.}

The paper is organized as follows.  Section~\ref{sec:tidal} introduces the tidal
deformation tensor, the T-web classification and the independent shear
invariants.  Section~\ref{sec:doro} establishes the Doroshkevich distribution,
the $\mathbb{Z}_2$ theorem, the Gaussian volume fractions $\FVC=7.96\%$ and the
Gaussian classification entropy.  Section~\ref{sec:diffentropy} computes the
differential entropy of the tidal field and proves the factor-of-five shear
information advantage.  Section~\ref{sec:fractal} develops the multifractal
geometry and the information dimension.  Section~\ref{sec:logdoro} derives the
Log-Doroshkevich distribution and shows that it preserves the $\mathbb{Z}_2$
symmetry.  Section~\ref{sec:matter} computes the matter fractions in the
Zel'dovich approximation and reconciles the Gaussian and evolved entropy
budgets.  Sections~\ref{sec:breaking}--\ref{sec:entropy} derive the
symmetry-breaking master formulae, the Shannon entropy differences and the
primordial non-Gaussianity signal.  Section~\ref{sec:redshift} works out the
redshift evolution of the entropy, the master relation linking it to the growth
rate, and the critical redshifts at which each mechanism operates.  We discuss
the results in Sec.~\ref{sec:discussion} and conclude in
Sec.~\ref{sec:conclusions}.  Two appendices collect the derivation of the
Doroshkevich distribution and its moments (Appendix~\ref{app:doro}) and the
closed-form $(u,v,t)$ integral representations of the volume and matter
fractions (Appendix~\ref{app:fractions}).

\subsection{Information theory}
\label{ssec:infotheory}

Information theory offers a complementary, basis-independent approach.  The
Shannon entropy~\cite{Shannon1948}
\begin{equation}
  H = -\sum_{i} p_i \log_2 p_i \quad \text{(bits)}
  \label{eq:shannon}
\end{equation}
quantifies the uncertainty, or information content, of a probability
distribution $\{p_i\}$.  For a \emph{continuous} random variable $\rho$, such as
a tidal eigenvalue or a shear invariant, the differential (continuous) Shannon
entropy generalizes Eq.~\eqref{eq:shannon} to
\begin{equation}
  h[\rho] = -\int \calP(\rho)\log\calP(\rho)\,d\rho .
  \label{eq:diffentropy}
\end{equation}
Unlike $H$, the differential entropy $h$ is not invariant under nonlinear
reparametrizations, but differences of differential entropies and mutual
information are invariant~\cite{Cover2006}.

Furthermore, the fractal geometry of the Cosmic Web introduces another layer of
information: each morphological component is characterized by a Hausdorff
dimension $D_H<3$~\cite{Mandelbrot1982,Martinez2002}, and the full multifractal
spectrum $g(\alpha)$ encodes how the singularity strengths of the density field
are distributed across scales.  The generalized R\'enyi dimensions $D_q$
interpolate between the Hausdorff dimension ($q\to0$), the information dimension
($q\to1$), and the correlation dimension ($q\to2$), providing a hierarchy of
information measures that bridges geometry and statistics.

\section{The Tidal Deformation Tensor and Morphological Classification}
\label{sec:tidal}

The gravitational tidal tensor is the Hessian of the peculiar gravitational
potential $\phi$ (related to $\delta$ by the Poisson equation
$\nabla^2\phi=4\pi G\bar\rho\,a^2\delta$ in comoving coordinates):
\begin{equation}
  \Pij(\vq) = \frac{\partial^2\phi}{\partial q_i\,\partial q_j}.
  \label{eq:tidalHessian}
\end{equation}
In linear theory $\Pij$ coincides (up to a sign convention) with the
deformation tensor of Lagrangian perturbation theory.  Let
$\lambda_1\ge\lambda_2\ge\lambda_3$ be the ordered eigenvalues of $\Pij$.  By
the Poisson equation,
\begin{equation}
  \delta = \Tr(\bm\Psi) = \lambda_1 + \lambda_2 + \lambda_3.
  \label{eq:trace}
\end{equation}

\subsection{T-web classification}
\label{ssec:tweb}

Following the T-web formalism of Hahn et al.~\cite{Hahn2007}, a point in the
density field is classified by the number $n_+$ of \emph{positive} eigenvalues,
relative to a threshold $\lth\ge0$:
\begin{equation}
  \text{web type} =
  \begin{cases}
    \text{cluster}  & n_+ = 3, \\
    \text{filament} & n_+ = 2, \\
    \text{wall}     & n_+ = 1, \\
    \text{void}     & n_+ = 0.
  \end{cases}
  \label{eq:Tweb}
\end{equation}
Each configuration corresponds to collapse along 3, 2, 1, or 0 principal axes,
respectively, in the Zel'dovich approximation.  The threshold $\lth$ is a free
parameter; $\lth=0$ is the standard choice in the linear regime and is adopted
throughout.

\subsection{Independent shear invariants}
\label{ssec:shear}

The three scalar invariants of $\Pij$ are
\begin{align}
  I_1 &= \Tr(\bm\Psi) = \lambda_1+\lambda_2+\lambda_3,\label{eq:I1}\\
  I_2 &= \half\!\left[(\Tr\bm\Psi)^2 - \Tr(\bm\Psi^2)\right]
        = \lambda_1\lambda_2 + \lambda_1\lambda_3 + \lambda_2\lambda_3,\label{eq:I2}\\
  I_3 &= \det(\bm\Psi) = \lambda_1\lambda_2\lambda_3.\label{eq:I3}
\end{align}
It is, however, more natural for information-theoretic purposes to work with the
{\em traceless} shear tensor
\begin{equation}
  T_{ij} = \Psi_{ij} - \frac{\delta}{3}\delta_{ij}^{(K)},
  \label{eq:sheardef}
\end{equation}
where $\delta_{ij}^{(K)}$ is the Kronecker delta.  The two non-trivial
invariants of $\mathbf{T}$ are
\begin{align}
  \mathcal{Q} &\equiv \Tr(\mathbf{T}^2)
               = \Tr(\bm\Psi^2) - \tfrac{1}{3}\delta^2,
  \label{eq:Q}\\
  \mathcal{A} &\equiv \Tr(\mathbf{T}^3)
               = \Tr(\bm\Psi^3) - \delta\,\Tr(\bm\Psi^2)
                 + \tfrac{2}{9}\delta^3.
  \label{eq:Acubic}
\end{align}
$\mathcal{Q}\ge0$ is the squared shear amplitude, vanishing only for perfectly
isotropic configurations ($\lambda_1=\lambda_2=\lambda_3$).  $\mathcal{A}$
measures the skewness of the eigenvalue distribution and changes sign between
prolate and oblate geometries.  The triplet $(\delta,\mathcal{Q},\mathcal{A})$
uniquely determines the set $\{\lambda_i\}$ up to permutation; it forms a
complete, irreducible basis for the information content of $\Pij$.

{
\section{The Doroshkevich Distribution and its $\mathbb{Z}_2$ Symmetry}
\label{sec:doro}

\subsection{The Doroshkevich PDF}
\label{ssec:doropdf}

In the linear regime, when the density perturbations form a homogeneous
isotropic Gaussian random field with tidal-tensor variance
$\sigma_{\rm DS}^2=\langle\Psi_{ii}^2\rangle=\sigma_\delta^2/5$, the joint PDF
of the three ordered tidal eigenvalues $(\lambda_1\ge\lambda_2\ge\lambda_3)$ was
derived by Doroshkevich~\cite{Doroshkevich1970,Doroshkevich1978,Shandarin1983}
(see Appendix~\ref{app:doro}):
\begin{eqnarray}
  dW &=& \frac{15^3}{8\pi\sqrt{5}\,\sigma_\delta^6}
       (\lambda_1-\lambda_2)(\lambda_1-\lambda_3)(\lambda_2-\lambda_3) \nonumber\\
       &&\hspace{3mm}\times\exp\!\left[-\frac{6I_1^2-15I_2}{2\sigma_\delta^2}\right]
       \dif\lambda_1\dif\lambda_2\dif\lambda_3\\
       &=& \frac{25\sqrt5}{6\pi\,\sig_\delta^6}
  \, v^{3/2} \exp\!\left[-\frac{5v}{2\sig_\delta^2}-\frac{u^2}{2\sig_\delta^2}\right]
  \dif v\,\dif u,
  \label{eq:Doro}
\end{eqnarray}
where $\sigma_\delta^2$ is the linear density variance smoothed on a given
scale~\cite{Bardeen1986}, $I_1=\Tr\bm\Psi$, $I_2=\frac13(I_1^2-\Tr\bm\Psi^2)$,
and $-\infty<\lambda_3\le\lambda_2\le\lambda_1<\infty$.  The Vandermonde-like
factor $\Delta\equiv(\lambda_1-\lambda_2)(\lambda_1-\lambda_3)(\lambda_2-\lambda_3)$
enforces the ordering and suppresses configurations with degenerate eigenvalues,
reflecting the tendency of tidal fields to favor \emph{triaxial} rather than
isotropic geometries.  The second line of Eq.~\eqref{eq:Doro} follows from the
exact factorization~\cite{Bardeen1986} under the change of variables
$u=I_1\sim\mathcal{N}(0,\sigma_\delta^2)$ and
$v=I_1^2-3I_2\sim\Gamma[5/2,2\sigma_\delta^2/5]$ (independent Gaussian and Gamma
variables, respectively), with $I_3$ uniform on its admissible range
$[I_3^-,I_3^+]$ at fixed $(u,v)$; the derivation is given in
Appendix~\ref{app:doro}.

\subsection{The $\mathbb{Z}_2$ theorem}
\label{ssec:Z2G}

The $\mathbb{Z}_2$ symmetry $\lambda_i\to-\lambda_i$ corresponds to $u\to-u$;
since $\mathcal{P}_u$ is symmetric, the transformation is measure-preserving and
maps the cluster domain $(+,+,+)$ onto the void domain $(-,-,-)$, and filaments
$(+,+,-)$ onto walls $(-,-,+)$, with equal weight.  This establishes:

\begin{theorem}[$\mathbb{Z}_2$ Theorem --- Gaussian]
  \label{thm:Z2G}
  For any isotropic Gaussian random field at $\lth=0$,
  $\FVC=\FVV$ and $\FVF=\FVW$ exactly.
\end{theorem}

\noindent
The integral-level proof, based on the root-reversal identity
$\mu_3(v,-t)=-\mu_1(v,t)$ in the $(u,v,t)$ parametrization, is given in
Appendix~\ref{app:fractions}.
}

\subsection{Web-component volume fractions}
\label{ssec:fractions}

The web-component volume fractions can be computed in two complementary ways,
yielding numerically consistent results.

\textit{Physical volume fractions.}
Monte Carlo sampling from the exact $6\times6$ tidal-tensor covariance matrix,
counting each spatial voxel equally, gives the fraction of \emph{volume} with a
given sign pattern.  With $N=10^7$ samples (machine precision $\sim0.01\%$):
\begin{equation}
  \FVC = \FVV = 7.96\%,
  \quad
  \FVF = \FVW = 42.04\%.
  \label{eq:fracA}
\end{equation}

\textit{Covariance structure.}
The fractions above use the isotropic covariance
$\langle\Psi_{ij}\Psi_{kl}\rangle=(\sigma_\delta^2/15)(\delta_{ij}\delta_{kl}
+\delta_{ik}\delta_{jl}+\delta_{il}\delta_{jk})$, which gives
$\langle\Psi_{ii}^2\rangle=\sigma_{\rm DS}^2\equiv\sigma_\delta^2/5$ and
$\langle\Psi_{ij}^2\rangle=\sigma_{\rm DS}^2/3$ for $i\neq j$.  With this
$1{:}1/3$ diagonal-to-off-diagonal variance ratio, the quadratic form reduces to
$Q=(u^2+5v)/(2\sigma_\delta^2)$, which is rotation-invariant, so the Vandermonde
factor $\Delta(\bm\lambda)$ cancels in the change of variables
$(\lambda_i)\to(u,v,t)$ (see Appendix~\ref{app:fractions}).  Sampling from this
marginal distribution reproduces $\FVC=7.96\%$, in agreement with the physical
$6\times6$ Monte Carlo and with the Shandarin--Doroshkevich--Zel'dovich
(SDZ1983)~\cite{Shandarin1983} value of $\approx8\%$.  The rest follow from the
$\mathbb{Z}_2$ symmetry and completion, $\FVV=\FVC=7.96\%$,
$\FVW=\FVF=0.5-\FVC=42.04\%$, such that $\sum_i F_i^{(V)}=1$.

\subsection{Entropy of the Gaussian classification}
\label{ssec:HGauss}

The discrete Shannon entropy of the T-web classification is
$H^{(V)}=-\sum_i F_i\log_2 F_i$, where $i\in\{C,F,W,V\}$.  With the physical
volume fractions~\eqref{eq:fracA},
\begin{equation}
  H^{(V)}_G = 1.632\;\mathrm{bits}
  \label{eq:HGauss}
\end{equation}
($81.6\%$ of the maximum $\log_2 4=2$ bits), with each component contributing
$h_{\rm C}=h_{\rm V}=0.291$ bits and $h_{\rm F}=h_{\rm W}=0.526$ bits
respectively.  The $\mathbb{Z}_2$ symmetry forces exact degeneracy between the
cluster/void pair and the filament/wall pair; any breaking of these equalities
is a direct observational signal, and is the subject of
Secs.~\ref{sec:breaking}--\ref{sec:entropy}.

\section{Differential Entropy of the Tidal Field}
\label{sec:diffentropy}

Equation~\eqref{eq:Doro} can be written compactly as
\begin{equation}
  \calP(u,v) =
    \mathcal{N}\, v^{3/2}
    \exp\!\left(-\frac{5v}{2\sig_\delta^2}-\frac{u^2}{2\sig_\delta^2}\right),
  \label{eq:DoroshkevichI}
\end{equation}
in terms of the independent variables $u=I_1\in(-\infty,\infty)$ and
$v=I_1^2-3I_2\in(0,\infty)$ (Appendix~\ref{app:doro}).  The marginal PDF of each
eigenvalue and the fraction of volume in each T-web class follow by integration,
as in Sec.~\ref{ssec:fractions}.

\subsection{Differential entropy of the eigenvalue distribution}
\label{ssec:hlambda}

The full differential Shannon entropy of the joint eigenvalue distribution is
\begin{equation}
  h[\bm\lambda] = -\iiint_{\lambda_1\geq\lambda_2\geq\lambda_3}
    \hspace{-2mm}\calP(\lambda)\,\log\calP(\lambda)\;d\lambda_1\,d\lambda_2\,d\lambda_3,
  \label{eq:hlambda}
\end{equation}
or equivalently
\begin{equation}
  h[\bm I] = -\iint_{I_1^2\geq 3I_2}\calP(I_1,I_2)
  \log\calP(I_1,I_2)\;dI_1\,dI_2.
  \label{eq:hI}
\end{equation}
Substituting Eq.~\eqref{eq:DoroshkevichI} and separating the Gaussian and
Vandermonde contributions, with $\sigma_\lambda^2\equiv\sigma_\delta^2/15$ the
one-dimensional eigenvalue variance,
\begin{equation}
  h[\bm I] = \log\left(8\pi\sqrt5\,\sigma_\lambda^6\right)
 + 3  - \frac{3}{2}\,\Big\langle\log v\Big\rangle,
  \label{eq:hlambda_expand}
\end{equation}
where $\langle\cdot\rangle$ denotes averaging over $\calP(u,v)$.  The first two
terms are determined by $\sigma_\lambda$ alone; the third encodes the
information in the \emph{anisotropy} of the eigenvalue spectrum.  Using the
moments of the Doroshkevich distribution~\cite{Lee2000} (Appendix~\ref{app:doro}),
\begin{equation}
  h[\bm I] = \log\left(\frac{8\pi\sqrt5\,\sigma_\lambda^6}{e}\right)
  + \frac{3}{2}\gamma - \frac{3}{2}\log\left(\frac{3}{2}\sigma_\lambda^2\right),
  \label{eq:hlambda_final}
\end{equation}
where $\gamma=0.5772\dots$ is the Euler constant.  The term
$-\tfrac32\langle\log v\rangle>0$ is positive definite and represents the
information \emph{lost} due to the ordering constraint; it increases with the
anisotropy of the tidal field and is minimized in isotropic configurations
($\Delta\to0$ or $I_1^2\to3I_2$), where the distribution concentrates on the
constraint surface.

\subsection{Entropy of the shear invariants}
\label{ssec:hshear}

Since $(\delta,\mathcal{Q},\mathcal{A})$ constitute a change of variables from
$(\lambda_1,\lambda_2,\lambda_3)$, the information content is conserved up to the
Jacobian.  In the Gaussian linear field, $\delta$ and $\mathcal{Q}$ are
\emph{statistically independent}~\cite{Pogosyan1998}: the isotropic part
(density) decouples from the anisotropic part (shear).  Consequently the entropy
decomposes as
\begin{equation}
  h[\bm\lambda] = h[\delta] + h[\mathcal{Q},\mathcal{A}] + \text{const},
  \label{eq:decomp}
\end{equation}
where the constant accounts for the Jacobian of the transformation.  The
differential entropy of the density contrast alone is that of a Gaussian,
\begin{equation}
  h[\delta] = \half\log\!\left(2\pi e\,\sigma_\delta^2\right).
  \label{eq:hgauss}
\end{equation}
The additional entropy residing in $(\mathcal{Q},\mathcal{A})$ therefore
captures \emph{geometric} information about anisotropic collapse that is
invisible to density-field statistics.  In practice $\mathcal{Q}$ follows a
scaled chi-squared distribution (a sum of squares of five Gaussian variates in
the irreducible representation of the shear), while the joint distribution of
$(\mathcal{Q},\mathcal{A})$ is related to that of a $3\times3$ Gaussian
orthogonal ensemble~\cite{Mehta2004}.

In the Doroshkevich model, the five independent components of the traceless
shear tensor $\mathbf{T}$ each contribute $\half\log(2\pi e\,\sigma_T^2)$ bits of
differential entropy, where $\sigma_T^2=\tfrac{2}{15}\sigma_\delta^2$, yielding a
total shear entropy
\begin{equation}
  h[\mathbf{T}] = \frac{5}{2}\log\!\left(2\pi e\,\sig_T^2\right)
  \simeq \frac{5}{2}\log\!\left(2\pi e\,\sig_\delta^2\right) - 5.
  \label{eq:hT}
\end{equation}
{Here $\sigma_T^2=\tfrac{2}{15}\sigma_\delta^2$ is the \emph{mean}
per-component shear variance; the five components are not identically distributed
({\em e.g.} $\langle T_{11}^2\rangle=\tfrac{4}{45}\sigma_\delta^2$,
$\langle T_{12}^2\rangle=\tfrac{1}{15}\sigma_\delta^2$), so $h[\mathbf{T}]$
strictly involves $\tfrac12\log\det\mathbf{C}_T$, of which Eq.~\eqref{eq:hT} is
the equal-variance approximation.  The factor of five is best read as a counting
of \emph{degrees of freedom}, five shear components against the single trace
$\delta$, rather than as a literal ratio of differential entropies: the latter
is not reparametrization-invariant and depends on the smoothing scale through the
additive constant $-5=\tfrac52\log(2/15)$ in Eq.~\eqref{eq:hT}.  Evaluating the
ratio explicitly,}
\begin{equation}
  {\frac{h[\mathbf{T}]}{h[\delta]}
  = 5 + \frac{5\log(2/15)}{\log(2\pi e\,\sigma_\delta^2)}
  \;\xrightarrow{\ \sigma_\delta^2\to\infty\ }\; 5,}
  \label{eq:hratio}
\end{equation}
{so the value $5$ is only approached formally.}

\begin{table}[h]
  \centering
  {
  \caption{Ratio of shear to density differential entropy as a function of
  smoothing scale.  The asymptotic value $5$ is reached only for
  $\sigma_\delta\gg1$, deep in the nonlinear regime where the Gaussian
  Doroshkevich form no longer applies; at $R=8\,h^{-1}$Mpc the two are
  comparable. (1 bit = $\ln 2$ nats.)}
  \label{tab:hratio}
  \begin{ruledtabular}
  \begin{tabular}{cccc}
    $\sigma_\delta^2$ & $h[\delta]$ (nats) & $h[\mathbf{T}]$ (nats) & ratio \\
    \hline
    $0.66\ (R{=}8\,h^{-1}{\rm Mpc})$ & $1.21$ & $1.02$ & $0.84$ \\
    $1.0$  & $1.42$ & $2.06$ & $1.45$ \\
    $4.0$  & $2.11$ & $5.52$ & $2.62$ \\
    $100$  & $3.72$ & $13.57$ & $3.65$ \\
  \end{tabular}
  \end{ruledtabular}}
\end{table}

{The invariant content is therefore the degree-of-freedom count: a local
density value retains one of the six tensor components.  Crucially, this
information is discarded only by a \emph{local} (one-point) density statistic;
the full density field determines the tidal tensor exactly through
$\Psi_{ij}(\mathbf{k})=(k_ik_j/k^2)\,\delta(\mathbf{k})$, and in the linear
Gaussian regime the tidal field carries no two-point information beyond the
power spectrum $P(k)$, which is then a sufficient statistic.  The genuine
information advantage of anisotropic statistics is thus a \emph{non-Gaussian},
nonlinear effect, maximised at intermediate-to-small scales and consistent with
the Fisher-information gains of Bonnaire et al.~\cite{Bonnaire2022}.}

\subsection{Mutual information with morphology}
\label{ssec:MI}

The mutual information between $\delta$ and the morphological class (T-web label)
is
\begin{equation}
  \mathcal{I}(\delta\,;\,\text{web}) = H(\text{web}) - H(\text{web}|\delta),
  \label{eq:MI}
\end{equation}
where $H(\text{web})$ is the classification entropy of Eq.~\eqref{eq:HGauss} and
$H(\text{web}|\delta)$ is the conditional entropy of the morphological label
given the local density, i.e.\ the one-point value $\delta(\mathbf{x})$
at a single smoothed position, as opposed to the full field
$\{\delta(\mathbf{x}')\}$, which fixes $\bm\Psi$ entirely.  
Because two regions with the same $\delta$ can have
different $(\mathcal{Q},\mathcal{A})$ and hence different T-web types,
$H(\text{web}|\delta)>0$: the shear invariants carry information about morphology
that is not captured by the density alone.  The mutual information between
$(\mathcal{Q},\mathcal{A})$ and the T-web class,
\begin{equation}
  \mathcal{I}(\mathcal{Q},\mathcal{A}\,;\,\text{web})
    = H(\text{web}) - H(\text{web}|\mathcal{Q},\mathcal{A}),
  \label{eq:MIshear}
\end{equation}
quantifies how much of the morphological entropy is resolved by the shear.

\section{Fractal Geometry and Multifractal Entropy}
\label{sec:fractal}

\subsection{Hausdorff dimensions of Cosmic Web components}
\label{ssec:hausdorff}

The Cosmic Web is not space-filling: each morphological component occupies a set
of non-integer Hausdorff dimension $D_H<3$~\cite{Mandelbrot1982}.  Measurements
from $N$-body simulations and galaxy surveys~\cite{Martinez2002,Peebles1980}
yield the approximate values
\begin{equation}
  D_H \approx
  \begin{cases}
    1.2\text{--}1.5 & \text{clusters (nodes)}, \\
    1.7\text{--}1.9 & \text{filaments}, \\
    2.4\text{--}2.7 & \text{walls}, \\
    3.0 & \text{voids (fill 3-D volume)},
  \end{cases}
  \label{eq:DH}
\end{equation}
where lower $D_H$ reflects stronger geometrical compression.  Clusters and
filaments have the lowest dimensions because matter is compressed along three
and two axes, respectively, while walls are compressed along one axis only.

\subsection{R\'enyi entropy and generalized dimensions}
\label{ssec:renyi}

The multifractal formalism~\cite{Hentschel1983,Halsey1986} generalises the
notion of a single fractal dimension to a spectrum of \emph{generalized
dimensions}
\begin{equation}
  D_q = \frac{1}{q-1}\lim_{\ell\to 0}
        \frac{\log \sum_i p_i(\ell)^q}{\log \ell}, \quad q \neq 1,
  \label{eq:Dq}
\end{equation}
where the sum is over cells of linear size $\ell$ and $p_i(\ell)$ is the
probability (mass fraction) in cell $i$.  The $q=0$ limit gives the Hausdorff
dimension, $D_0=D_H$.  The $q\to1$ limit defines the \emph{information dimension}
\begin{equation}
  D_1 = \lim_{\ell\to 0}
        \frac{-\sum_i p_i(\ell)\log p_i(\ell)}{\log(1/\ell)}
      = \lim_{\ell\to 0}\frac{H(\ell)}{\log(1/\ell)},
  \label{eq:D1}
\end{equation}
which is precisely the rate at which the Shannon entropy $H(\ell)$ of the density
field grows as the resolution scale $\ell$ decreases.  The $q=2$ limit gives the
correlation dimension $D_2$, related to the two-point correlation function,
$\xi(r)\sim r^{-(3-D_2)}$~\cite{Alonso2013,Martinez2002}.

The generalized dimensions satisfy $D_0\ge D_1\ge D_2\ge\cdots$, with equality
holding only for a monofractal (uniform mass distribution on a single fractal
set).  The degree of multifractality is quantified by the \emph{width} of the
singularity spectrum $g(\alpha)$, the Legendre transform of $(q-1)D_q$:
\begin{equation}
  g(\alpha) = q\alpha - (q-1)D_q, \quad \alpha = \frac{d[(q-1)D_q]}{dq}.
  \label{eq:falpha}
\end{equation}
Here $\alpha$ is the local H\"older (Lipschitz) exponent of the density measure,
and $g(\alpha)$ is the Hausdorff dimension of the set of points with that
exponent.  The peak of $g(\alpha)$ at $\alpha=D_1$ corresponds to the dominant
singularity strength, and the width $\Delta\alpha=\alpha_\mathrm{max}-\alpha_\mathrm{min}$
measures the heterogeneity of the fractal.

\subsection{Entropy growth rate and information dimension}
\label{ssec:infodim}

Combining Eqs.~\eqref{eq:D1} and~\eqref{eq:shannon}, the information dimension of
a given web component $i=\{$voids, walls, filaments, clusters$\}$ is
\begin{equation}
  D_1^{(i)} = \lim_{\ell\to 0}
    \frac{H^{(i)}(\ell)}{\log(1/\ell)},
  \label{eq:D1component}
\end{equation}
where $H^{(i)}(\ell)$ is the Shannon entropy of the mass distribution
\emph{within} component $i$ at resolution $\ell$.  For a uniform distribution on
a $D_H$-dimensional fractal, $H^{(i)}(\ell)\approx D_H\log(1/\ell)$ and hence
$D_1^{(i)}=D_H^{(i)}$.  Departures from this monofractal scaling, measured by
$D_0-D_1>0$, quantify the excess information due to the heterogeneous,
multifractal distribution of matter within each component.

For filaments, the relatively low $D_H\approx1.8$ combined with the large mass
fraction $p_F\approx0.4$ implies a high \emph{information density per unit
volume}: the filament network packs a disproportionate amount of Shannon entropy
into a small fraction of the cosmic volume.  This is the geometric analogue of
the maximum-entropy result of Sec.~\ref{ssec:budget}.

{
\section{The Log-Doroshkevich Distribution}
\label{sec:logdoro}

\subsection{Physical motivation}
\label{ssec:logmotiv}

In the mildly nonlinear regime, the density field is well approximated by a
lognormal~\cite{Coles1991}: $1+\delta=e^g$, where
$g\sim\mathcal{N}(\mu_g,\sigma_g^2)$, $\mu_g=-\sigma_g^2/2$, and
$\sigma_g^2=\ln(1+\sigma_\delta^2)$.  There are two distinct physical contexts
in which an eigenvalue PDF for a lognormal density field is needed:

\textit{Scenario A: Gaussian displacement, lognormal density.}
In the Zel'dovich approximation the displacement field $s_i(\mathbf{q})$ is
Gaussian, and the deformation tensor eigenvalues $(\alpha,\beta,\gamma)$ of
$\partial s_i/\partial q_k$ follow the Doroshkevich PDF exactly.  The density
arises from the Jacobian
$(1+\delta)=[(1-D(z)\alpha)(1-D(z)\beta)(1-D(z)\gamma)]^{-1}$ and is
approximately lognormal; the eigenvalues themselves remain Gaussian.  No new PDF
is required.

\textit{Scenario B: Tidal tensor of a lognormal field.}
The density itself is $1+\delta=e^g$.  The gravitational potential satisfies
$\nabla^2\phi\propto\delta$, so in Fourier space
$T_{ij}(\mathbf{k})=-(k_ik_j/k^2)\hat\delta(\mathbf{k})$.  Since
$\hat\delta(\mathbf{k})=\widehat{e^g-1}$ involves a nonlinear convolution,
$T_{ij}$ is not a Gaussian random matrix and no exact Doroshkevich-type PDF
exists in general.

\textit{Scenario C: Local lognormal --- the Log-Doroshkevich PDF.}
In the local (pointwise) approximation, the tidal eigenvalues of the lognormal
field are related to those of the Gaussian log-field by the monotone map
$\lambda_i=e^{\nu_i}-1$.  This is the standard assumption underlying the
lognormal model of Ref.~\cite{Coles1991} and gives the tractable, exact PDF
derived below.

\subsection{Derivation}
\label{ssec:logderive}

Apply the change of variables $\lambda_i=e^{\nu_i}-1$ (equivalently
$\nu_i=\ln(1+\lambda_i)$) to the Gaussian Doroshkevich PDF for eigenvalues
$\nu_i$ of the log-density field $g$ with variance $\sigma_g^2$.  Since $e^x$ is
strictly monotone:
\begin{itemize}
\item \textit{Ordering}: $\nu_1\ge\nu_2\ge\nu_3\Leftrightarrow
  \lambda_1\ge\lambda_2\ge\lambda_3$.
\item \textit{Support}: $\nu_i\in(-\infty,\infty)\to\lambda_i\in(-1,\infty)$.
  The hard lower bound $\lambda_i>-1$ ($\delta>-1$) is automatic.
\item \textit{Signs}: $\nu_i\gtrless0\Leftrightarrow\lambda_i\gtrless0$
  (since $e^0=1$ and $\lambda_i(0)=0$).
\item \textit{Jacobian}: $\dif\nu_i=\dif\lambda_i/(1+\lambda_i)$.
\item \textit{Vandermonde}: $\nu_i-\nu_j=\ln[(1+\lambda_i)/(1+\lambda_j)]$.
\item \textit{Invariants}: the log-field invariants are
  $J_1\equiv I_1^{(g)}=\sum_i\ln(1+\lambda_i)$ and
  $J_2\equiv I_2^{(g)}=\sum_{i<j}\ln(1+\lambda_i)\ln(1+\lambda_j)$.
\end{itemize}
Substituting into Eq.~\eqref{eq:Doro} with $\sigma\to\sigma_g$ yields the
\emph{Log-Doroshkevich distribution}:
\begin{multline}
  dW^{\rm LN}(\lambda_1,\lambda_2,\lambda_3) = \mathcal{N}_{\mathrm{LN}}
    \frac{\displaystyle\prod_{i<j}\ln\!\frac{1+\lambda_i}{1+\lambda_j}}
         {\displaystyle\prod_i(1+\lambda_i)}\\
    \times\exp\!\left[-\frac{3}{2\sigma_g^2}(2J_1^2-5J_2)\right]
    \prod_i\dif\lambda_i,
  \label{eq:logDoro}
\end{multline}
with support $\lambda_1\ge\lambda_2\ge\lambda_3>-1$ and
$\sigma_g^2=\ln(1+\sigma_\delta^2)$.  The normalization $\mathcal{N}_{\rm LN}$ is
inherited from the Gaussian case; the PDF integrates to unity by construction.

\subsection{Key properties}
\label{ssec:logprop}

Table~\ref{tab:comparison} summarizes how the Log-Doroshkevich PDF differs from
the Gaussian original.

\begin{table}[h]
  \caption{Comparison of the Gaussian Doroshkevich and Log-Doroshkevich
  eigenvalue PDFs.}
  \label{tab:comparison}
  \begin{ruledtabular}
  \begin{tabular}{lll}
    Feature & Gaussian Doroshk. & Log-Doroshkevich \\
    \hline
    Support & $(-\infty,\infty)$ & $(-1,\infty)$\\
    Variance parameter & $\sigma_\delta^2$ & $\sigma_g^2=\ln(1+\sigma_\delta^2)$\\
    Vandermonde & $\lambda_i-\lambda_j$ & $\ln\frac{1+\lambda_i}{1+\lambda_j}$\\
    Jacobian factor & 1 & $\prod_i(1+\lambda_i)^{-1}$\\
    Overdense tail & Gaussian & Power-law: $\lambda_1^{-5/\sigma_g^2}$\\
    Wall at $\lambda_i\to-1$ & Finite (unphysical) & $\to0$ (physical)\\
    Linear limit & -- & $\to$ Gauss. Doroshk.\\
    Volume fractions & $\FVC=\FVV$ & $\FC^{\mathrm{LN}}=\FV^{\mathrm{LN}}$\\
  \end{tabular}
  \end{ruledtabular}
\end{table}

\textit{Linear limit.}  For $\sigma_g\to0$: $\ln(1+\lambda_i)\to\lambda_i$,
$(1+\lambda_i)^{-1}\to1$, and $J_k\to I_k$.  The Log-Doroshkevich reduces exactly
to the Gaussian Doroshkevich.

\textit{Power-law tail.}  As $\lambda_1\to\infty$: $J_1\sim\ln\lambda_1$, so the
exponent scales as $(\ln\lambda_1)^2$ and the PDF decays as
$\lambda_1^{-5/\sigma_g^2}$, a power law rather than Gaussian.  Clusters can
have arbitrarily high density, but their probability weight is algebraically
small.

\textit{Hard wall.}  As $\lambda_3\to-1^+$: the Jacobian denominator
$(1+\lambda_3)^{-1}\to\infty$ and the Vandermonde
$\ln[(1+\lambda_2)/(1+\lambda_3)]\to+\infty$.  The product
$[\ln\cdots]/(1+\lambda_3)$ vanishes because the logarithm grows slower than the
pole, so $dW^{\mathrm{LN}}\to0$ at the boundary.  Voids approaching the
completely empty limit are increasingly improbable.

\textit{Skewness.}  The reduced skewness of the lognormal density field
$S_3^{\mathrm{LN}}=3+\sigma_\delta^2+O(\sigma_\delta^4)$, approaching the
leading-order value of 3.  Compare $S_3^{\mathrm{grav}}=34/7+\neff\approx3.4$--$4.7$:
the lognormal model captures approximately $70$--$90\%$ of the gravitational
skewness with the correct support.

\subsection{The $\mathbb{Z}_2$ theorem for the lognormal}
\label{ssec:Z2LN}

The most important result of this section is:

\begin{theorem}[$\mathbb{Z}_2$ Theorem --- Lognormal]
  \label{thm:Z2LN}
  The Log-Doroshkevich PDF satisfies
  $\FC^{\mathrm{LN}}=\FV^{\mathrm{LN}}$ and
  $\FF^{\mathrm{LN}}=\FW^{\mathrm{LN}}$ exactly.
\end{theorem}

\begin{proof}
The volume fraction in each domain is the integral of $dW^{\mathrm{LN}}$ over the
corresponding sign region.  Since $\nu_i=\ln(1+\lambda_i)$ is a strictly
monotone bijection with $\nu_i>0\Leftrightarrow\lambda_i>0$ and
$\nu_i<0\Leftrightarrow\lambda_i<0$, the sign domains of $\{\lambda_i\}$ are in
exact one-to-one correspondence with those of $\{\nu_i\}$.  The measure
$dW^{\mathrm{LN}}$ is the image of $dW$ under this bijection, so the integrals
over corresponding domains are equal.  Since $dW$ satisfies
Theorem~\ref{thm:Z2G}, so does $dW^{\mathrm{LN}}$.
\end{proof}

This theorem establishes that the $\mathbb{Z}_2$ symmetry is \emph{not broken} by
the lognormal transformation at the volume-fraction level, regardless of the
amplitude $\sigma_\delta$.  The breaking appears only in density-weighted
integrals.

\textit{Density-weighted breaking.}
The mean density in each web domain is
$\langle1+\delta\rangle_i=\langle e^{J_1}\rangle_i$, where
$J_1=\ln[(1+\lambda_1)(1+\lambda_2)(1+\lambda_3)]$.  Under
$\lambda_i\to-\lambda_i$ (the $\mathbb{Z}_2$ map), $J_1\to-J_1$, so
$e^{J_1}\to e^{-J_1}\neq e^{J_1}$.  The density-weighting factor breaks the
symmetry.  From Monte Carlo ($\sigma_\delta=0.5$, $N=5\times10^6$):
\begin{equation}
  \frac{\langle1+\delta\rangle_{\rm C}}{\langle1+\delta\rangle_{\rm V}}
  = \frac{\langle e^{J_1}\rangle_{\rm C}}{\langle e^{J_1}\rangle_{\rm V}}
  \approx 4.8,
  \label{eq:densityasym}
\end{equation}
confirming that cluster regions are far denser than void regions are empty, even
in the lognormal model without any gravitational evolution.

\subsection{Validity of the pointwise map}
\label{ssec:pointwise}

Theorem~\ref{thm:Z2LN} is a statement about Scenario~C, the pointwise map
$\lambda_i=e^{\nu_i}-1$, and not about the tidal tensor of a physical lognormal
field (Scenario~B).  The distinction matters, because the two are equivalent
only in the linear limit.  The pointwise map holds the eigen\emph{vectors} fixed
and rescales the eigen\emph{values} by a monotone, sign-preserving function;
this is precisely why it preserves the sign structure and hence $\mathbb{Z}_2$
exactly.  The physical construction, by contrast, is the composition of a
pointwise nonlinear step $g\mapsto\delta=e^{g}-1$ with the \emph{nonlocal}
operator $T_{ij}=\partial_i\partial_j\nabla^{-2}\delta$.  The second step
rebuilds the tidal tensor from the transformed density by integrating over the
surrounding field, so the principal axes of $T_{ij}[\delta]$ are rotated
relative to those of $T_{ij}[g]$ and $\lambda_i[\delta]\neq e^{\nu_i}-1$.

More fundamentally, a lognormal field is not invariant under $\delta\to-\delta$:
it is skewed, with a hard floor at $\delta=-1$ and a long overdense tail.  Since
the tidal tensor is \emph{linear} in $\delta$ through the Poisson equation, the
map $\delta\to-\delta$ does send $\lambda_i\to-\lambda_i$, but the underlying
measure is not symmetric, so the eigenvalue PDF acquires an odd part at first
order in the field amplitude, of order $S_3^{\rm LN}\sigma_\delta$ with
$S_3^{\rm LN}\simeq3$.  The volume-fraction $\mathbb{Z}_2$ symmetry is therefore
broken already in Scenario~B, at linear order in $\sigma_\delta$.

We have verified this directly.  Generating a Gaussian field $g$ with a
power-law spectrum, forming $1+\delta=e^{g-\sigma_g^2/2}$, computing the exact
tidal tensor by fast Fourier transform, and diagonalizing it at each grid point,
we find that the cluster and void volume fractions split as
$\FC/\FV-1\simeq-\sigma_\delta$ at small amplitude (averaged over realizations,
with the Gaussian case $\delta=g$ reproducing $\FC=\FV$ to numerical precision).
The sign is the physical one: a positively skewed, void-dominated field places
more volume in all-negative (void) tidal configurations than in all-positive
(cluster) ones, in line with the hierarchy seen in $N$-body
classifications~\cite{Cautun2014}.  This first-order breaking is the same
skewness effect quantified for gravitational collapse in
Sec.~\ref{sec:breaking}, with $S_3^{\rm LN}$ in place of $S_3^{\rm grav}$; the
exact $\mathbb{Z}_2$ of the Log-Doroshkevich PDF survives only because the
pointwise model discards precisely this term.  The Log-Doroshkevich distribution
should accordingly be understood as a tractable nonlinear PDF with the correct
support and overdense tail, not as a claim that the physical tidal field retains
an exact volume-fraction symmetry beyond linear order.
}

\section{Matter Fractions and the Entropy Budget}
\label{sec:matter}

The Gaussian volume fractions of Sec.~\ref{ssec:fractions} are
$\mathbb{Z}_2$-symmetric and fix the maximum-symmetry entropy $H^{(V)}_G=1.632$
bits.  The observed Cosmic Web is not symmetric: in $N$-body and survey
classifications voids fill most of the volume while clusters occupy almost none.
This section bridges the two.  Density weighting and gravitational evolution
transfer matter from voids into filaments and clusters, breaking the
$\mathbb{Z}_2$ degeneracy and driving the configuration-space entropy down from
its Gaussian baseline.  We first compute the matter fractions in the Zel'dovich
approximation, then assemble the full entropy budget and show how the Gaussian
and evolved fractions are two epochs of the same system.

\subsection{Matter fractions in the Zel'dovich approximation}
\label{ssec:matterfrac}

The \emph{matter fraction} in web component $i$ is the density-weighted volume
fraction.  In the Zel'dovich approximation, where $D\equiv D(z)$ is the linear
growth factor ($D(0)=1$) and the eigenvalues $\lambda_i$ are fixed at their
present-day linear values,
$(1+\delta)=[(1-D\lambda_1)(1-D\lambda_2)(1-D\lambda_3)]^{-1}$, giving
\begin{equation}
  f_i^{(M)} = \frac{\int_{\mathcal{D}_i}(1+\delta)\,P(\bm\lambda)\,d^3\lambda}
               {\int_\mathrm{all}(1+\delta)\,P(\bm\lambda)\,d^3\lambda},
  \label{eq:fM}
\end{equation}
where $\mathcal{D}_i$ is the sign domain of component $i$ and $P(\bm\lambda)$ is
the physical Doroshkevich PDF.  Shell crossing occurs when $D(z)\lambda_1=1$, so
$D(z)\,\sigma_{\rm DS}$ (where $\sigma_{\rm DS}=\sigma_\delta/\sqrt5$) is the
natural nonlinearity parameter; at $z=0$, $R=8\,h^{-1}$Mpc it equals
$\sigma_8/\sqrt5=0.363$.  At $D(z)=0$, $1+\delta=1$, recovering the volume
fractions.  At linear order in $D$, $1+\delta\approx1+D\,I_1$, giving
\begin{equation}
  f_i^{(M)}\approx F_i^{(V)}\bigl(1 + D\langle I_1\rangle_i\bigr),
  \label{eq:fMlinear}
\end{equation}
where $\langle I_1\rangle_i\equiv\langle I_1\rangle_{\mathcal{D}_i}/F_i^{(V)}$ is
the \emph{conditional} mean of the trace within domain $\mathcal{D}_i$.  From
Monte Carlo ($N=5\times10^6$, $\sigma_{\rm DS}=1$),
\begin{eqnarray}
  \langle I_1\rangle_{\rm C} &=& -\langle I_1\rangle_{\rm V} = +3.72\,\sigma_{\rm DS},
  \\[1mm]
  \langle I_1\rangle_{\rm F} &=& -\langle I_1\rangle_{\rm W} = +1.23\,\sigma_{\rm DS}.
  \label{eq:I1means}
\end{eqnarray}
The $\mathbb{Z}_2$ antisymmetry
$\langle I_1\rangle_C+\langle I_1\rangle_V=0$ (verified to $0.03\%$) implies that
$\fMC-\FVC=-(\fMV-\FVV)$ to first order in $D(z)\,\sigma_{\rm DS}$, so clusters
gain and voids lose exactly the same matter content.  The $\mathbb{Z}_2$ symmetry
$\fMC=\fMV$ is thus already broken at first order in the density weighting.

\begin{table}[h]
  \caption{Matter fractions from the Zel'dovich approximation,
    Eq.~\eqref{eq:fM}, computed by Monte Carlo with the physical $6\times6$ tidal
    covariance ($N=5\times10^6$, $\sigma_{\rm DS}=1$).  Shell-crossed
    trajectories ($D\lambda_1\ge1$) are excluded; their fraction is listed in
    the last column.  For reference, the NEXUS+~\cite{Cautun2014} values at $z=0$
    are $\fMC\approx11\%$, $\fMF\approx46\%$, $\fMW\approx36\%$, $\fMV\approx7\%$.}
  \label{tab:matter}
  \begin{ruledtabular}
  \begin{tabular}{lccccr}
    $D(z)\sigma_{\rm DS}$ &  $\fMC$ & $\fMF$ & $\fMW$ & \hspace{-3mm}$\fMV$
              & \hspace{-3mm}shell-cross \\
    \hline
    0.00 &  8.0\% & 42.0\% & 42.0\% &  8.0\% &  0.0\% \\
    0.10 & 11.6\% & 46.7\% & 36.3\% &  5.4\% &  0.0\% \\
    0.20 & 18.1\% & 49.8\% & 28.8\% &  3.3\% &  0.0\% \\
    0.25 & 30.0\% & 47.0\% & 21.1\% &  1.9\% &  0.1\% \\
    0.30 & 45.8$\pm$11\% & 43.7$\pm$11\% &  9.7$\pm$2.2\% &  0.8$\pm$0.1\% &  0.7\% \\
  \end{tabular}
  \end{ruledtabular}
\end{table}

\subsection{Nonlinear evolution and shell crossing}
\label{ssec:shellcross}

Table~\ref{tab:matter} and Fig.~\ref{fig:evolution} give the matter fractions as
a function of redshift, for increasing $D(z)\,\sigma_{\rm DS}$, excluding
shell-crossed trajectories ($D\lambda_1\ge1$).  As $D(z)\,\sigma_{\rm DS}$ grows,
matter evacuates voids and accumulates in filaments and clusters.  At
$D(z)\,\sigma_{\rm DS}\approx0.2$ the cluster matter fraction (mildly nonlinear
regime) reaches $\sim18\%$.

The qualitative evolution proceeds in two stages.  In the \emph{quasi-linear}
stage ($D\sigma_{\rm DS}\lesssim0.2$, i.e.\ $z\gtrsim1$ at $R=8\,h^{-1}$Mpc) the
density weighting $(1+\delta)$ is a smooth, monotonic reweighting of the
underlying Gaussian eigenvalue distribution.  Matter drains steadily from voids
($\fMV:8.0\%\to3.3\%$) into clusters ($\fMC:8.0\%\to18.1\%$) and filaments
($\fMF:42.0\%\to49.8\%$), while walls lose ground ($\fMW:42.0\%\to28.8\%$).  No
trajectory has yet crossed $D\lambda_1=1$, so the Zel'dovich map remains
single-valued and invertible.

\begin{figure}
    \centering
    \includegraphics[width=0.99\linewidth]{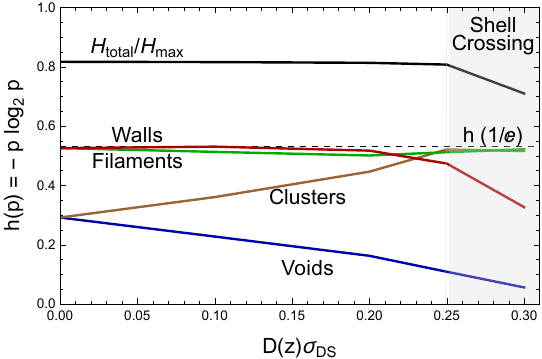}
    \caption{Evolution with redshift, $D(z)\sigma_{\rm DS}$, of the Shannon
    subentropies $h(p)$ of the four web components (color lines matching dots in
    Fig.~\ref{fig:HDS}), together with the relative fraction of total Shannon
    entropy to maximal, $H_{\rm total}/H_{\rm max}$ (black line).  After shell
    crossing the total Shannon entropy drops steeply, driven by the wall and void
    components, while clusters and filaments remain stable close to the
    subentropy peak $h(1/e)$ (dashed line).}
    \label{fig:evolution}
\end{figure}

The entry at $D\sigma_{\rm DS}=0.25$ ($z\approx0.74$ at $R=8\,h^{-1}$Mpc) marks
the threshold between the two stages.  Here the first multistream regions appear
($0.1\%$ of trajectories have $D\lambda_1\ge1$), and the cluster matter fraction
has climbed to $\fMC\approx30\%$, having overtaken the wall fraction and
approached the filament fraction ($\fMF\approx47\%$).  This value is significant:
$30\%$ sits just below $1/e\approx37\%$, the location of the maximum of the
subentropy $h(p)=-p\log_2 p$.  Just beyond this point the cluster fraction
crosses $p=1/e$ from below, reaching $\fMC\approx46\%$ by
$D\sigma_{\rm DS}=0.30$, while the filament fraction, having peaked near $50\%$
at $D\sigma_{\rm DS}=0.2$, descends toward the same region from above
($\fMF\approx44\%$ at $D\sigma_{\rm DS}=0.30$).  The two thus \emph{converge}
near the subentropy peak, driven by the diverging density weight that triggers
shell crossing.  Because $p=1/e$ is where a component contributes most to the
total entropy, the clustering of \emph{two} components there keeps the
cluster--filament part of the entropy budget high even as the walls and voids
are depleted; the net entropy decline between $D\sigma_{\rm DS}=0.25$ and $0.30$
(the last two panels of Fig.~\ref{fig:HDS}) is therefore moderate,
$H:1.62\to1.42$ bits, rather than catastrophic.

\begin{figure}
    \centering
    \includegraphics[width=0.49\linewidth]{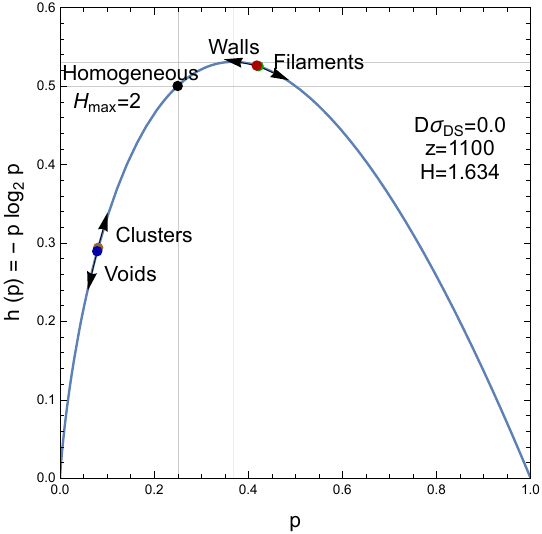}\\
    \includegraphics[width=0.49\linewidth]{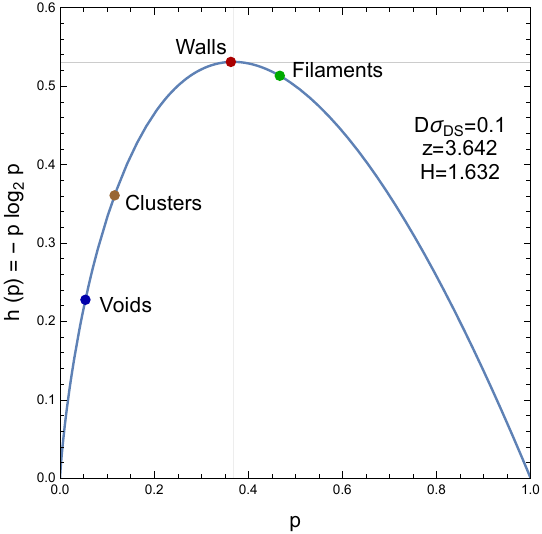}
    \includegraphics[width=0.49\linewidth]{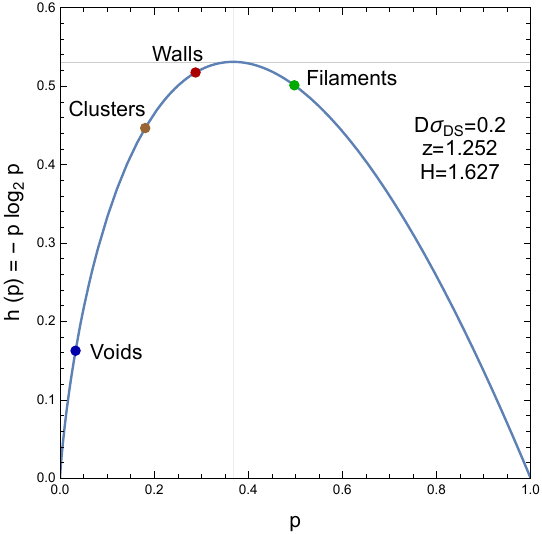}
    \includegraphics[width=0.49\linewidth]{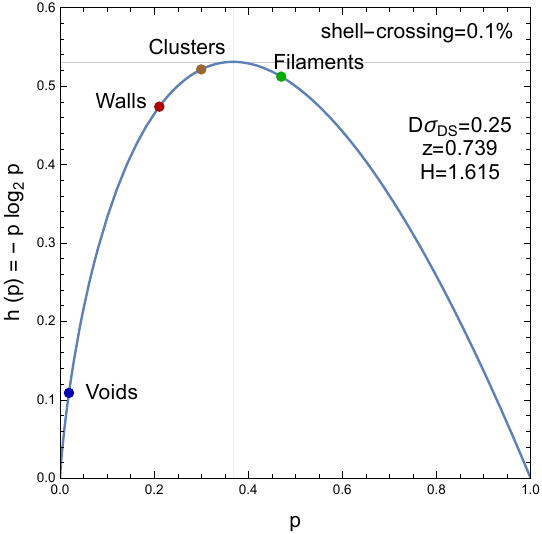}
    \includegraphics[width=0.49\linewidth]{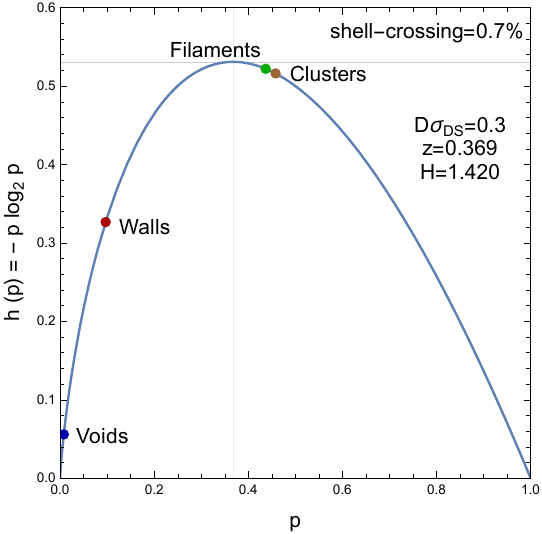}
    \caption{The location of the Cosmic Web components in the Shannon subentropy
    plane $(p, h)$.  {\em Top panel.}  A homogeneous distribution has equal
    probabilities for the four web components and thus maximal entropy,
    $H_{\rm max}=\log_2 4=2$.  A Gaussian density field gives a lower entropy
    than maximal ignorance (i.e.\ it carries more information),
    $H_{\rm GDF}=1.632$.  {\em Lower panels.}  As the density field becomes more
    ordered under gravitational collapse, the entropy decreases further and the
    Cosmic Web becomes more structured, $H_{\rm CW}\simeq1.6$.  After shell
    crossing the entropy decreases more steeply, $H_{SC}=1.42$, driven by the
    depletion of voids and walls, while clusters and filaments converge near the
    subentropy peak (see Fig.~\ref{fig:evolution}).}
    \label{fig:HDS}
\end{figure}

The \emph{shell-crossing} stage sets in once $D\sigma_{\rm DS}$ approaches
$\sim0.3$.  At $D\sigma_{\rm DS}=0.30$ ($z\approx0.37$ at $R=8\,h^{-1}$Mpc),
$0.7\%$ of trajectories have $D\lambda_1\ge1$ and are excluded as multistream
regions where the single-stream density diverges.  The density weight
$(1+\delta)=\prod_i(1-D\lambda_i)^{-1}$ grows steeply as $D\lambda_1\to1^-$ and
pumps mass into the densest cells, so the cluster matter fraction climbs to
$\fMC\approx46\%$ and overtakes the filaments, which fall to $\fMF\approx44\%$.
The walls and voids are drained in turn ($\fMW\approx10\%$, $\fMV\approx0.8\%$).
Clusters and filaments thus end up nearly equally probable, both sitting just
above the subentropy peak at $p=1/e$, while walls and voids slide toward
$p\to0$.  This is the matter-space signature of the onset of nonlinear collapse:
a rapid transfer of mass out of the underdense regions and into the densest
ones.  We caution that the fractions in this regime are sensitive to the
diverging density weight near the pole $D\lambda_1=1$, so the precise values at
$D\sigma_{\rm DS}=0.30$ carry sizeable Monte Carlo uncertainty and should be read
as indicative of the trend rather than as precise predictions.

\subsection{The evolved entropy budget}
\label{ssec:budget}

The Shannon subentropy plane of Fig.~\ref{fig:HDS} makes the information content
of the classification visible at a glance: the total Shannon entropy
$H=\sum_i h(F_i)$ is the sum of the four ordinates $h(p)=-p\log_2 p$.  At
$D\sigma_{\rm DS}=0$ (the linear/Gaussian limit, formally $z=1100$) the field has
$H=1.632$ bits, with clusters and voids degenerate at $(p,h)=(0.080,0.292)$ and
filaments and walls at $(0.420,0.526)$, the exact pairwise degeneracy
enforced by the $\mathbb{Z}_2$ symmetry.  As gravitational collapse proceeds
($D\sigma_{\rm DS}=0.1,0.2$) the cluster-void and filament-wall pairs split apart
and slide along the curve, but the total entropy barely moves ($H=1.632,1.627$
bits): the entropy gained by clusters moving toward the peak is almost exactly
compensated by voids moving down the left branch, a near-cancellation that is the
first-order $\mathbb{Z}_2$-antisymmetry of Eq.~\eqref{eq:fMlinear} expressed in
entropy space.  The behavior changes at $D\sigma_{\rm DS}=0.25$ and $0.3$, once
shell crossing begins: the total entropy falls visibly, to $H=1.615$ and then
$H=1.420$ bits, a $13\%$ reduction from the Gaussian value, as the
distribution becomes more ordered.

Carried to the fully nonlinear, present-day regime, this evolution produces the
strongly asymmetric fractions measured in $N$-body simulations and galaxy
surveys.  Table~\ref{tab:fractions} lists the evolved volume and mass filling
fractions of the four T-web components, together with their contributions to the
discrete Shannon entropy.  These are the broken-symmetry endpoint of the
$\mathbb{Z}_2$-symmetric Gaussian initial condition of
Table~\ref{tab:matter}: voids, which fill half the Gaussian volume, have grown
to fill $\sim77\%$ of the present-day volume, while clusters, equal to voids in
the Gaussian field, occupy only $\sim1\%$.  The volume-weighted entropy has
fallen from $H^{(V)}_G=1.632$ bits to $H^{(V)}\approx1.03$ bits.

\begin{table}[h]
  \centering
  \caption{Evolved (present-day) volume and mass filling fractions of the four
  Cosmic Web components, their density-contrast range, Hausdorff dimension, and
  matter-weighted entropy contributions.  These are the broken-symmetry endpoint
  of the $\mathbb{Z}_2$-symmetric Gaussian fractions of Sec.~\ref{ssec:fractions}.
  Filling fractions from Refs.~\cite{Cautun2014,Springel2005,Springel2018};
  fractal dimensions from Refs.~\cite{Martinez2002,Cautun2014}.{\ These
  evolved fractions are not universal: they depend on the cosmology, redshift,
  smoothing scale, grid resolution and, most strongly, on the web
  classifier and the eigenvalue threshold $\lth$~\cite{ForeRomero2009}.  Even
  within a single reference, different identification methods yield discrepant
  fractions; the values here are representative rather than definitive (see
  text).}}
  \label{tab:fractions}
  \begin{ruledtabular}
  \begin{tabular}{lccccc}
    Type & $p_i^{(V)}$ & $p_i^{(M)}$ & $\delta$ range &
    $D_H$ & $-p_i^{(M)}\log_2 p_i^{(M)}$ \\
    \hline
    Voids     & 0.77 & 0.03 & $<-0.8$   & 3.0  & 0.152 \\
    Walls     & 0.16 & 0.33 & $0$--$5$  & 2.5  & 0.528 \\
    Filaments & 0.06 & 0.40 & $5$--$50$ & 1.8  & 0.529 \\
    Clusters  & 0.01 & 0.24 & $>100$    & 1.2  & 0.494 \\
  \end{tabular}
  \end{ruledtabular}
\end{table}

Summing the contributions of the evolved fractions,
\begin{align}
  H^{(V)} &= 1.03 \;\text{bits}, \\
  H^{(M)} &= 1.70 \;\text{bits}.
\end{align}
The maximum possible entropy for a four-component system is
$H_\mathrm{max}=\log_2 4=2$ bits, achieved when all components carry equal mass.
The matter-weighted entropy thus operates at $\sim85\%$ of the maximum
information capacity, while the volume-weighted entropy, dominated by the single
overwhelming void component, has fallen to $\sim52\%$.

{Two qualifications are in order, since the underlying fractions depend on
the classifier and smoothing scale (Table~\ref{tab:fractions}).  First, the
$\sim85\%$ figure is representative, not universal.  Because the per-component
entropy $H(p)=-p\log_2 p$ is flat near its maximum at $p^\star=e^{-1}$, moderate
shifts in the fractions move $H^{(M)}$ only at second order: re-evaluating the
budget across published fraction sets spanning $\lth=0$ to $\sim0.2$--$0.4$ and
different finders (T-web, \textsc{nexus}+, watershed) keeps $H^{(M)}$ in the
range $\simeq1.5$--$1.8$ bits ($\sim75\%$--$90\%$ of capacity), with the
ordering --- filaments and walls near the per-component maximum, voids and
clusters on the low-entropy tails --- preserved throughout.  It is this
ordering, not the precise number, that is robust.  Second, the $15\%$ deficit
$H_\mathrm{max}-H^{(M)}=0.30$ bits does \emph{not} reside in the density
contrast.  It is the redundancy of the four-symbol morphological distribution,
formally the Kullback--Leibler divergence of the mass fractions from the uniform
distribution,
\begin{equation}
  H_\mathrm{max}-H^{(M)}
  = \sum_i p_i^{(M)}\log_2\!\frac{p_i^{(M)}}{1/4}
  = D_{\rm KL}\!\bigl(p^{(M)}\,\|\,\text{uniform}\bigr),
  \label{eq:KLdeficit}
\end{equation}
which measures only how unequally matter is partitioned among the four classes.
This discrete budget ($\sim85\%$ of a four-symbol capacity) and the
degree-of-freedom budget of Sec.~\ref{ssec:hshear} (the density retains one of
six continuous tensor components) are distinct measures living in different
spaces; they are not complementary shares of a single total and must not be
added.}

\begin{figure}
    \centering
    \includegraphics[width=0.99\linewidth]{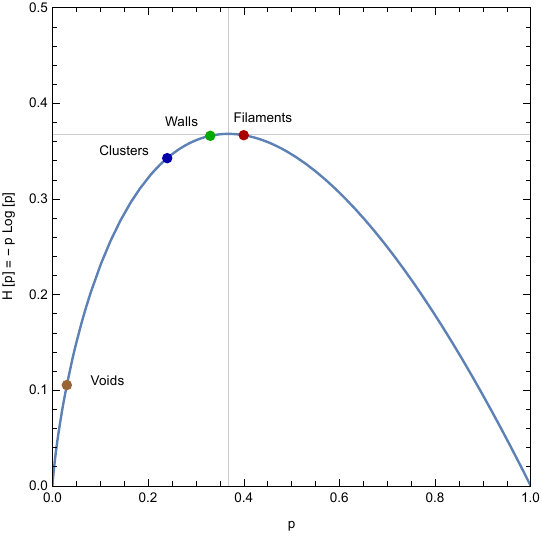}
    \caption{The function $H(p)=-p\log p$ peaks at $p=1/e\approx0.37$.
    Filaments and walls hold $\sim$40\% and 33\% of the cosmic matter,
    respectively, placing them near the theoretical maximum.  Voids and
    clusters lie on the low-entropy tails: too rare or too dominant.}
    \label{fig:Shannon}
\end{figure}

\subsection{Dominant contribution from filaments}
\label{ssec:filaments}

The entropy function {$H(p)=-p\log_2 p$} attains its global maximum at
$p^\star=e^{-1}\approx0.368$, with {$H(p^\star)$}$=(\log_2 e)\,e^{-1}\approx0.531$
bits.  The matter fractions of filaments and walls place them near this maximum
(Fig.~\ref{fig:Shannon}).  Consequently, \emph{filaments and walls carry the
maximum possible per-component entropy} for their matter fraction, and dominate
the total $H^{(M)}$.  Combined with the low Hausdorff dimension of filaments
($D_H\approx1.8$, Sec.~\ref{ssec:hausdorff}), this implies the highest
information density per unit volume of any web component: the filament network
packs a disproportionate amount of Shannon entropy into a small fraction of the
cosmic volume.  The two measures resolve complementary aspects of the density
field, the \emph{volume}-weighted maximum selects mildly underdense regions
at the void--wall boundary, while the \emph{mass}-weighted maximum selects the
filament regime~\cite{Bernardeau1994,Bernardeau2002}.

{
\section{Breaking of the $\mathbb{Z}_2$ Symmetry by Gravitational Collapse}
\label{sec:breaking}

The exact $\mathbb{Z}_2$ symmetry $\FC=\FV$ established in Sec.~\ref{sec:doro}
holds only for a zero-mean Gaussian random field.  Gravitational evolution
breaks it through a well-defined hierarchy of perturbative corrections.  The
primary effect is the generation of a non-zero reduced skewness
$S_3\equiv\langle\delta^3\rangle/\sigma^4$ by second-order gravitational
instability.  Below we show that this skewness shifts the cluster and void
fractions by equal and opposite amounts, derive the geometric coefficient that
controls the shift, and explain the origin and scale dependence of the total
reduced skewness $\Stot(R)$.

\subsection{The total reduced skewness $\Stot(R)$}
\label{ssec:skewness}

The one-point statistics of the smoothed density field
$\delta_R(\mathbf{x})=\int W_R(\mathbf{x}-\mathbf{y})\,\delta(\mathbf{y})\,d^3y$
(with a top-hat window of radius $R$) are characterized at leading non-Gaussian
order by the reduced skewness
\begin{equation}
  S_3(R) = \frac{\langle\delta_R^3\rangle}{\langle\delta_R^2\rangle^2}
          = \frac{\langle\delta_R^3\rangle}{\sigma^4(R)}.
\end{equation}
Two physically distinct mechanisms contribute in the quasi-linear regime, adding
linearly in the total reduced skewness:
\begin{equation}
  \Stot(R) = \Sgrav(R) + \Spng(R).
  \label{eq:S3tot}
\end{equation}

\textit{Gravitational contribution $\Sgrav(R)$.}
Second-order perturbation theory generates a coupling
$\delta^{(2)}\propto(\delta^{(1)})^2$ that sources a positive skewness in the
density field~\cite{Bernardeau1994,Bernardeau2002}.  In the large-scale limit the
reduced skewness of a top-hat-filtered field evaluates to the well-known
tree-level result
\begin{equation}
  \Sgrav(R) = \frac{34}{7} + \neff(R),
  \label{eq:S3grav}
\end{equation}
where $\neff(R)=d\ln\sigma^2(R)/d\ln R$ is the local effective spectral index of
the linear power spectrum at scale $R$.  The universal constant
$34/7\approx4.86$ arises from the angular integrals of the second-order kernel
$F_2(\mathbf{k}_1,\mathbf{k}_2)$ for gravitational collapse, and is independent of
the power-spectrum shape.  The $\neff$ term (which ranges from $-3$ to $0$ as $R$
varies from small to large scales) encodes the shape correction from the
smoothing window.  For $R=8\,h^{-1}$Mpc in flat $\Lambda$CDM the effective
spectral index is $\neff\approx-1.5$, giving $\Sgrav(8)\approx3.36$.  The
physical origin of the positive sign is the asymmetry of gravitational collapse:
overdense regions contract to arbitrarily high densities, while underdense
regions can evacuate only to the hard floor $\delta=-1$.  The resulting density
PDF develops a long positive tail and a truncated negative tail.

\textit{PNG contribution $\Spng(R)$.}
For local-type primordial non-Gaussianity (PNG), the curvature perturbation is
$\Phi=\phi+\fnl(\phi^2-\langle\phi^2\rangle)$ where $\phi$ is a Gaussian
field~\cite{Maldacena2003,Verde2000}.  Projecting the resulting local bispectrum
$B_\Phi=2\fnl[P_\phi(k_1)P_\phi(k_2)+\mathrm{perm.}]$ through the matter transfer
function and the top-hat window gives the primordial third
moment~\cite{Verde2000,LoVerde2008}
\begin{widetext}
\begin{equation}
  \langle\delta_R^3\rangle_{\rm PNG}
  = 6\fnl\!\int\!\frac{\dif^3k_1\,\dif^3k_2}{(2\pi)^6}\,
    \mathcal{M}(k_1)\mathcal{M}(k_2)\mathcal{M}(k_{12})\,
    P_\phi(k_1)P_\phi(k_2)\,
    W_R(k_1)W_R(k_2)W_R(k_{12}),
  \label{eq:d3PNG}
\end{equation}
\end{widetext}
with $\mathcal{M}(k)=\tfrac{2}{3}k^2T(k)/(\Omega_m H_0^2)$ the matter transfer
operator and $k_{12}=|\mathbf{k}_1+\mathbf{k}_2|$, so that
\begin{eqnarray}
  \Spng(R)=\frac{\langle\delta_R^3\rangle_{\rm PNG}}{\sigma^4(R)}
          &\equiv& S_3^{(1)}(R)\,\fnl, \label{eq:S3PNG}
  \\ {\rm with} \ \ 
  S_3^{(1)}(R=8\,h^{-1}{\rm Mpc})&\simeq&3\times10^{-4},
  \nonumber
\end{eqnarray}
in agreement with the calibrated value of Grossi et al.~\cite{Grossi2009}. 
Unlike the gravitational skewness, $\Spng$ is \emph{not} an $O(1)$ number.  
The second-order relation $\delta^{(2)}\sim(\delta^{(1)})^2$ makes
both $\langle\delta^3\rangle_{\rm grav}$ and $\sigma^4$ scale as $\sigma^4$, so
$\Sgrav$ is a pure number; the primordial third moment instead scales as $\fnl$
times the transfer-weighted integral~\eqref{eq:d3PNG} and therefore inherits the
smallness of the primordial potential fluctuations ($\sigma_\Phi\sim10^{-5}$),
yielding $S_3^{(1)}\sim10^{-4}$ rather than $O(1)$.  Both $\Sgrav$ and $\Spng$ are
redshift-independent functions of the smoothing scale $R$.

\textit{The expansion parameter.}
The amplitude of non-Gaussian corrections to the density PDF is controlled by the
single dimensionless number
\begin{eqnarray}
  \eps(R,z)&\equiv&\Stot(R)\,\sigma(R,z), \label{eq:eps3}
  \\[2mm]
  {\rm with} \ \ \sigma(R,z)&=&\sigma_0(R)\,D(z),
  \nonumber
\end{eqnarray}
the product of the (redshift-independent) reduced skewness and the
(redshift-dependent) rms density variance.  A perturbative description requires
$|\eps|\ll1$.  For the gravitational contribution alone at $R=8\,h^{-1}$Mpc,
$\eps=2.726\,D(z)$, which exceeds unity at $z\lesssim5$, confirming that the
quasi-linear treatment applies to the high-redshift universe ($z\gtrsim5$) and
must be supplemented by the Log-Doroshkevich PDF (Sec.~\ref{sec:logdoro}) at
lower redshifts.  On smaller scales the variance is larger; on larger scales
$\Sgrav$ is larger because $\neff$ is less negative.  The optimum scale for a
perturbative analysis is therefore $R\sim2$--$8\,h^{-1}$Mpc at $z\sim3$--$10$.
For the PNG contribution,
$\eps^{\rm PNG}=\Spng\,\sigma\simeq3\times10^{-4}\fnl\,\sigma_0(R)D(z)\ll1$ for any
realistic $\fnl$, so the perturbative treatment of the primordial term is never in
question; its limitation is instead its smallness, quantified in
Sec.~\ref{ssec:DHpng}.

\subsection{Edgeworth expansion of the density PDF}
\label{ssec:edgeworth}

In the quasi-linear regime~\cite{Bernardeau2002,Kitaura:2010tr}, the corrections
to the Gaussian one-point PDF from a small skewness are captured by the Edgeworth
series.  Retaining only the leading skewness term,
\begin{equation}
  P(\delta)=P_G(\delta)\!\left[1+\frac{\eps}{6}H_3\!
  \left(\frac{\delta}{\sigma}\right)\right]
  +O(\eps^2),
  \label{eq:Edgeworth}
\end{equation}
where $H_3(x)=x^3-3x$ is the third Hermite polynomial and $P_G(\delta)$ is the
Gaussian PDF with variance $\sigma^2=\sigma^2(R,z)$.  The Hermite polynomial
$H_3$ appears because the third cumulant of the distribution is
$\kappa_3=S_3\sigma^4=\eps\,\sigma^3$; the Edgeworth expansion is the Fourier dual
of the cumulant generating function, and $H_3$ is the eigenfunction of the
Gaussian measure associated with the third cumulant.  The normalization and mean
are preserved to this order.  Since $H_3$ is odd, the PDF is skewed: the positive
tail is enhanced relative to the negative tail when $\eps>0$.  This is precisely
the signal of gravitational collapse.

\subsection{Cluster and void fraction shifts}
\label{ssec:shifts}

The cluster volume fraction at fixed $R$ and $z$ is
$\FC(R,z)=\int_{x_c\sigma}^\infty P(\delta)\,\dif\delta$, where the threshold
$x_c\sigma$ is defined implicitly by the Gaussian value $\FVC=F_G=7.96\%$, giving
$x_c=\Phi^{-1}(1-F_G)=1.408$ ($\Phi$ the standard normal CDF).  Substituting
Eq.~\eqref{eq:Edgeworth},
\begin{equation}
  \FVC(R,z) = F_G + \frac{\eps}{6}\int_{x_c}^\infty H_3(\nu)\,\phi(\nu)\,\dif\nu,
\end{equation}
where $\nu=\delta/\sigma$ and $\phi(\nu)=e^{-\nu^2/2}/\sqrt{2\pi}$.  The integral
evaluates exactly via the Hermite recurrence
$(\dif/\dif\nu)[H_2(\nu)\phi(\nu)]=-H_3(\nu)\phi(\nu)$,
\begin{equation}
  \int_{x_c}^\infty H_3(\nu)\,\phi(\nu)\,\dif\nu = H_2(x_c)\,\phi(x_c)
  = (x_c^2-1)\,\phi(x_c).
  \label{eq:HermInt}
\end{equation}
Defining the geometric coefficient
\begin{equation}
  A \equiv \frac{(x_c^2-1)\,\phi(x_c)}{6} = 0.02423,
  \label{eq:Acoeff}
\end{equation}
the shifted fractions are
\begin{align}
  \FVC(R,z) &= F_G + A\,\eps,
  \label{eq:FCshift}\\
  \FVV(R,z) &= F_G - A\,\eps.
  \label{eq:FVshift}
\end{align}
The void result follows from the same integral over $(-\infty,-x_c)$ together
with the oddness of $H_3$.  The coefficient $A>0$ because $x_c=1.408>1$, so
$x_c^2-1=0.981>0$ and $\phi(x_c)>0$.  Physically, the cluster fraction $\FVC$
increases and the void fraction $\FVV$ decreases by the \emph{same} amount: the
$\mathbb{Z}_2$ symmetry is broken in a balanced way, with the overdense sites
gained by clusters coming directly from voids.  The total fraction
$\FVC+\FVV\approx2\times7.96\%$ is preserved to first order in $\eps$.

\textit{Filaments and walls.}
The analogous computation for filaments uses the threshold
$x_f=\Phi^{-1}(\FVC+\FVF)$ separating walls from filaments, which satisfies
$x_f\approx0$ since $\FVC+\FVF\approx50\%$.  At $x_f\approx0$,
$A_\mathrm{fw}=(x_f^2-1)\phi(x_f)/6\approx-\phi(0)/6=-0.0665$ (negative, since
$x_f^2-1<0$).  Thus $\FVF<\FVW$ when $\eps>0$: positive skewness drains filaments
and replenishes walls as mass cascades from large (filament) to small (cluster)
scales.  The near-cancellation $\FVF+\FVW\approx2\times42.04\%$ is also preserved.

\subsection{Asymmetry ratios and their physical content}
\label{ssec:asymratios}

Dividing the cluster shift by the Gaussian baseline $F_G$,
\begin{align}
  \Rcv-1 &\equiv \frac{\FVC-\FVV}{F_G}
    \simeq \frac{2A}{F_G}\,\eps
    = 0.6088\,\eps(R,z),
  \label{eq:Rcv}\\
  \Rfw-1 &\equiv \frac{\FVF-\FVW}{F_G^\mathrm{fw}}
    \simeq \frac{2A_\mathrm{fw}}{F_G^\mathrm{fw}}\,\eps
    = -0.3164\,\eps(R,z).
  \label{eq:Rfw}
\end{align}
The numerical coefficients $2A/F_G=0.6088$ and
$2A_\mathrm{fw}/F_G^\mathrm{fw}=-0.3164$ are purely geometric: they depend only
on the Gaussian volume fractions $F_G$ and $F_G^\mathrm{fw}$ and on the shape of
the Doroshkevich PDF through $x_c$ and $x_f$.  They are independent of the power
spectrum, the cosmological model, and the redshift.  All cosmological information
enters exclusively through the factor $\eps(R,z)$ in Eq.~\eqref{eq:eps3}.  Both
ratios grow proportionally to $\sigma(R,z)\propto D(z)$ and vanish as
$z\to\infty$, returning to the Gaussian symmetry $\Rcv=\Rfw=1$ in the linear
limit.

\section{Shannon Entropy Differences and Primordial Non-Gaussianity}
\label{sec:entropy}

\subsection{Cluster-void entropy difference}
\label{ssec:DHcv}

The per-component entropies are $h_i=-F_i\log_2 F_i$, with $H=\sum_i h_i$.  The
function $h(p)=-p\log_2 p$ has derivative $h'(p)=-\log_2(ep)$, positive for
$p<1/e\approx37\%$ and negative for $p>1/e$.  Since the cluster and void
fractions satisfy $\FC=\FV\approx8\%\ll1/e$, their per-component entropies are
increasing functions of their probabilities, and the $\mathbb{Z}_2$ symmetry
$\FC=\FV$ implies $h_C=h_V$ exactly in the Gaussian case; any measured difference
is a direct signal of symmetry breaking.  An important distinction should be
drawn between the per-component entropy $h_i=F_i(-\log_2 F_i)$ and the
information content of a single observation of type $i$, which is $-\log_2 F_i$
bits.  The latter is the \emph{surprise} of finding one particular void (or
cluster); it grows as voids become rarer.  The former is the product of frequency
and surprise; since $F_{\rm V}<1/e$, $h_{\rm V}$ is an increasing function of
$F_{\rm V}$.  As gravitational collapse makes clusters more abundant and voids
rarer, the cluster per-component entropy $h_{\rm C}$ increases and the void
entropy $h_{\rm V}$ decreases: it is \emph{clusters}, not voids, that accumulate
Shannon entropy as structure forms~\cite{Pandey2019}.

Substituting the Edgeworth-corrected fractions $\FVC=F_G+A\eps$ and
$\FVV=F_G-A\eps$, Eqs.~\eqref{eq:FCshift}--\eqref{eq:FVshift}, into the
per-component entropy and Taylor expanding to first order in
$\eps=\Stot\sigma(R,z)$,
\begin{align}
  h_{\rm C} &= h(F_G+A\eps)\simeq h(F_G)+A\eps\,h'(F_G),
  \label{eq:hC}\\
  h_{\rm V} &= h(F_G-A\eps)\simeq h(F_G)-A\eps\,h'(F_G).
  \label{eq:hV}
\end{align}
Since $h'(F_G)=-\log_2(e\,F_G)=2.208>0$, the cluster entropy $h_{\rm C}$
increases and the void entropy $h_{\rm V}$ decreases from their common Gaussian
value $h(F_G)=0.291$ bits.  The signed difference
$\dH_{\rm CV}\equiv h_{\rm C}-h_{\rm V}$ satisfies
\begin{equation}
  \dH_{\rm CV} = 2A\,h'(F_G)\,\eps = 0.1070\,\eps > 0,
  \label{eq:DHcv_signed}
\end{equation}
and consequently
\begin{equation}
  \dH_{\rm CV}(R,z) = 0.1070\;\Stot(R)\,\sigma_0(R)\,D(z)\;\text{bits.}
  \label{eq:DHcv}
\end{equation}
The coefficient $2Ah'(F_G)=0.1070$ bits is a pure geometric number determined by
$x_c$ and $F_G$: it sets the sensitivity of the T-web entropy to the
non-Gaussian skewness.  The single factor $\eps$ captures all the cosmological
dependence, on the power-spectrum shape through $\Stot(R)$, on the
normalization through $\sigma_0(R)$, and on the redshift through $D(z)$.

\textit{Filament-wall entropy difference.}
The same argument applies to the filament-wall pair.  Since both $F_{\rm F}$ and
$F_{\rm W}$ exceed $1/e$ marginally and the filament-wall threshold
$x_f\approx0$ places the boundary close to $p=1/2$, the derivative
$h'(F_{\rm FW})$ is much smaller in magnitude, giving
\begin{equation}
  |\dH_{\rm FW}|(R,z) = 0.025\;\Stot(R)\,\sigma_0(R)\,D(z)\;\text{bits,}
  \label{eq:DHfw}
\end{equation}
about four times smaller than the cluster-void difference.  This reflects the
lower rarity of filaments and walls: $h(p)$ is nearly flat near
$p\approx42\%$, so the filament-wall entropy is much less sensitive to the
fraction shifts induced by the skewness.  The sign $\dH_{\rm FW}=h_{\rm F}-h_{\rm W}<0$
follows from $A_{\rm fw}<0$: positive skewness depletes filaments, so $h_{\rm F}$
decreases.

Table~\ref{tab:DH} gives the numerical values of both $\dH_{\rm CV}$ and
$\Rcv-1$ at $R=8\,h^{-1}$Mpc over a range of redshifts, for the gravitational
contribution alone ($\fnl=0$).  At $z=0$ the cluster-void entropy difference
reaches 0.292 bits, a $22\%$ deviation from the Gaussian value per component, but the Edgeworth approximation has already broken down at low redshifts, so
the table entries at $z\lesssim3$ should be taken as indicative order-of-magnitude
estimates.  In the high-redshift regime $z\gtrsim5$ ($\eps\lesssim0.32$), the
expansion is reliable and $\dH_{\rm CV}\lesssim0.034$ bits.

\begin{table}[h]
  \caption{Cluster-void entropy difference $\dH_{\rm CV}$ (bits) and asymmetry
    ratio $\Rcv-1$ as a function of redshift at $R=8\,h^{-1}$Mpc
    ($\sigma_0=0.812$, $\Stot\approx3.36$), for the gravitational contribution
    alone ($\fnl=0$).  The Edgeworth expansion is reliable for $\sigma\lesssim0.3$
    ($z\gtrsim5$); lower-redshift values are indicative.}
  \label{tab:DH}
  \begin{ruledtabular}
  \begin{tabular}{ccccc}
    $z$ & $D(z)$ & $\sigma(8,z)$ & $\dH_{\rm CV}$ (bits) & $\Rcv-1$ \\
    \hline
    0   & 1.000 & 0.812 & 0.292 & 1.658 \\
    1   & 0.612 & 0.497 & 0.178 & 1.015 \\
    3   & 0.319 & 0.259 & 0.093 & 0.529 \\
    5   & 0.214 & 0.173 & 0.062 & 0.354 \\
   10   & 0.117 & 0.095 & 0.034 & 0.193 \\
   20   & 0.061 & 0.050 & 0.018 & 0.101 \\
   50   & 0.025 & 0.020 & 0.007 & 0.042 \\
  100   & 0.0127 & 0.0103 & 0.0037 & 0.0211 \\
  \end{tabular}
  \end{ruledtabular}
\end{table}

\subsection{PNG contribution to the entropy difference}
\label{ssec:DHpng}

The total skewness $\Stot(R)=\Sgrav(R)+\Spng(R)$ enters Eq.~\eqref{eq:DHcv}
linearly.  The gravitational contribution produces a baseline entropy difference
present in any $\Lambda$CDM universe with the correct power spectrum.  Primordial
non-Gaussianity of local type adds $\Spng(R)=S_3^{(1)}(R)\,\fnl$ with
$S_3^{(1)}(8\,h^{-1}{\rm Mpc})\simeq3\times10^{-4}$, shifting the entropy
difference by
\begin{eqnarray}
  |\dH_{\rm PNG}|(R,z)
  &=& 0.1070\;\Spng(R)\,\sigma_0(R)\,D(z) \nonumber \\
  &\simeq& 3.2\times10^{-5}\,|\fnl|\;\sigma_0(R)\,D(z)\;\text{bits.}
  \label{eq:DHpng}
\end{eqnarray}
This is a genuine but very small effect.  Relative to the gravitational baseline
$|\dH_{\rm grav}|$, the PNG term is suppressed by
$\Spng/\Sgrav\approx10^{-4}|\fnl|$, so at the one-point level it is subdominant for
all realistic $|\fnl|$.  The entropy difference is a \emph{real-space}, \emph{scalar}
observable evaluated at a chosen smoothing scale $R$; it requires no Fourier
analysis and is insensitive to shot noise and survey geometry at the level of the
one-point PDF.  But --- unlike the scale-dependent bias signature
$\Delta b(k)\propto\fnl/k^2$, which encodes the primordial signal in the two-point
function on the largest scales --- it does not capture the leading LSS
manifestation of local PNG.  We therefore present $\dH_{\rm PNG}$ as a consistency
diagnostic rather than a competitive $\fnl$ estimator.

\begin{table}[h]
  \caption{Primordial contribution $|\dH_\mathrm{PNG}|$ to the cluster-void entropy
    difference at $z=10$, $R=8\,h^{-1}$Mpc ($\sigma=0.095$; gravitational
    baseline $|\dH_\mathrm{grav}|=0.034$ bits).  The PNG term is smaller than the
    baseline by $\sim10^{-4}$ for the $\fnl$ values shown.}
  \label{tab:fNL}
  \begin{ruledtabular}
  \begin{tabular}{ccc}
    $\fnl$ & $|\dH_\mathrm{PNG}|$ (bits) & ratio to baseline \\
    \hline
    $\pm1$   & $3.0\times10^{-6}$ & $9\times10^{-5}$ \\
    $\pm10$  & $3.0\times10^{-5}$ & $9\times10^{-4}$ \\
    $\pm100$ & $3.0\times10^{-4}$ & $9\times10^{-3}$ \\
  \end{tabular}
  \end{ruledtabular}
\end{table}

\subsection{Redshift-independence and the scale-separation strategy}
\label{ssec:scalesep}

The ratio of the PNG entropy contribution to the gravitational baseline is
\begin{equation}
  \frac{|\dH_\mathrm{PNG}|}{|\dH_\mathrm{grav}|}
  = \frac{\Spng(R)}{\Sgrav(R)}
  \approx 10^{-4}\,|\fnl|,
  \label{eq:ratio}
\end{equation}
redshift-independent (both $|\dH_\mathrm{grav}|$ and $|\dH_\mathrm{PNG}|$ share the
factor $\sigma_0(R)\,D(z)$, which cancels in the ratio) but, for any realistic
$|\fnl|$, negligibly small.  Two consequences follow.  First,
\emph{redshift tomography alone cannot separate $\fnl$ from the gravitational
background}: a measurement of $\dH_{\rm CV}$ at two redshifts $z_1$ and $z_2$ with
the same smoothing scale $R$ gives $|\dH(z_1)|/|\dH(z_2)|=D(z_1)/D(z_2)$, which
constrains the growth factor $D(z)$ but not $\fnl$ or $\Stot(R)$ separately.
Second, the primordial term is so far below the gravitational one that the
multi-scale strategy below cannot yield $\fnl$ constraints competitive with
scale-dependent bias; we retain it only to show, in principle, how scale
dependence distinguishes the primordial and gravitational skewness.

The key diagnostic is instead \emph{scale dependence}: the gravitational reduced
skewness $\Sgrav(R)=34/7+\neff(R)$ and the PNG reduced skewness
$\Spng(R)=S_3^{(1)}(R)\,\fnl$ carry different scale slopes, with
$\partial\Stot/\partial\ln R=\partial\neff/\partial\ln R
+\fnl\,\partial S_3^{(1)}/\partial\ln R$.  Because
$\partial S_3^{(1)}/\partial\ln R\sim10^{-4}$, the PNG slope is negligible against
the gravitational one for all realistic $\fnl$, so the derivative of the entropy
difference is $\partial\dH_{\rm CV}/\partial\ln R\propto\partial\neff/\partial\ln R$
to very high accuracy.
A measurement of $\dH_{\rm CV}$ at two smoothing scales $R_1$ and $R_2$ with known
$\neff(R_1)$ and $\neff(R_2)$ (from the linear power spectrum) determines both
$\Stot$ and $\fnl$ via
\begin{eqnarray}
  \fnl &\simeq&
  \frac{\dH_{\rm CV}(R_1)-\dH_{\rm CV}(R_2)}
       {C\cdot(\Spng(R_1)-\Spng(R_2))} \nonumber \\
       && - \frac{\Sgrav(R_1)-\Sgrav(R_2)}{\Spng(R_1)-\Spng(R_2)},
  \label{eq:fnl}
\end{eqnarray}
where $C=0.1070\,\sigma_0(R_{\rm ref})\,D(z)$ is a common calibration factor at a
reference scale $R_{\rm ref}$.  This multi-scale strategy is the T-web analogue of
the multi-tracer approach to scale-dependent bias, and is directly applicable to
three-dimensional reconstructions from forthcoming spectroscopic surveys such as
DESI~\cite{DESI2024PNG} and Euclid~\cite{Euclid2024}, which will map the web type
of millions of voxels out to $z\sim2$ on scales $R\sim2$--$20\,h^{-1}$Mpc.

\subsection{Astrophysical systematics}
\label{ssec:systematics}

Because $\dH_{\rm CV}$ is sourced by the third moment of the inferred tidal
field, any process that contributes to that moment competes with the tiny
primordial term $\Spng=S_3^{(1)}\fnl$.  Writing
$\Stot=\Sgrav+\Spng+S_3^{\rm RSD}+S_3^{\rm bias}$, an unmodeled systematic
skewness $\Delta S_3$ propagates to the inferred non-Gaussianity as
$\Delta\fnl=\Delta S_3/S_3^{(1)}\approx\Delta S_3/(3\times10^{-4})$.  Because the
RSD and bias skewnesses below are themselves $O(0.1$--$1)$, this ratio makes plain
that a one-point tidal statistic cannot isolate a local-PNG signal of realistic
amplitude: any astrophysical skewness of order $S_3^{\rm bias}\sim0.5$ overwhelms
the primordial term unless $|\fnl|\gtrsim10^3$.  Three effects dominate the
budget.

\textit{Redshift-space distortions.}  At linear (Kaiser) order, the
redshift-space map $\delta_s(\mathbf{k})=(1+f(z)\mu^2)\delta(\mathbf{k})$ is a
deterministic, mode-by-mode rescaling that keeps a Gaussian field Gaussian and
preserves the $\delta\to-\delta$ symmetry.  Since $\dH_{\rm CV}$ is precisely the
$\mathbb{Z}_2$-breaking observable, linear RSD shifts the overall fractions and
the total entropy but sources \emph{no} cluster--void asymmetry to leading order.
The genuine contamination is the second-order redshift-space skewness
$S_3^{\rm RSD}(f,\neff)$, set by the growth rate $f(z)=\Omega_m(z)^\gamma$ (known
from the background) and with its own scale dependence.  Its amplitude is
$|S_3^{\rm RSD}|\sim0.3\,f(z)\,\Sgrav\sim0.5$--$0.9$; even modeled to $\sim20\%$ it
leaves $\Delta S_3\sim0.1$--$0.2$, which through the factor $1/S_3^{(1)}$ maps to
$\Delta\fnl\sim(3$--$7)\times10^2$, far larger than any interesting value.

\textit{Nonlinear and tidal bias.}  The tidal tensor is reconstructed from biased
tracers, $\delta_g=b_1\delta+\tfrac12 b_2\delta^2+b_{s}s^2+\dots$, so the galaxy
skewness carries $S_3^{\rm bias}\simeq3b_2/b_1^2+(\text{tidal }b_{s}\text{
terms})$.  The $3b_2/b_1^2$ piece is a nearly scale-independent skewness directly
degenerate with $\Spng$ at a single scale: for $b_1\sim1.5$, $b_2\sim0.3$--$1$ it
reaches $S_3^{\rm bias}\sim0.4$--$0.9$, i.e.\ $\sim10^3$ times the primordial
$\Spng$ at $|\fnl|=1$, so even a $30\%$ uncertainty on $b_2$ swamps the signal
($\Delta\fnl\sim10^2$).  This is the same $\fnl$--$b_2$ degeneracy that limits
bispectrum analyses, here in an acute form because $S_3^{(1)}$ is so small.

\textit{Assembly bias.}  The most dangerous term for a tidal statistic is the
tidal bias $b_{s}$, which contaminates the reconstruction at the level of the
tensor itself rather than its trace.  Halo assembly bias renders $b_2$ and $b_s$
scale-dependent and dependent on unobserved halo properties, and is itself tied to
the hierarchical structure of the cosmic web~\cite{Coloma-Nadal:2024ccw}; a
scale-dependent $b_{s}(R)$ can mimic the mild scale slope of
$\Spng(R)=S_3^{(1)}(R)\,\fnl$ that the multi-scale fit of
Eq.~\eqref{eq:fnl} relies on to isolate $\fnl$.  Since that slope is itself of
order $10^{-4}$, an astrophysical $b_{s}(R)$ removes the lever arm of the
scale-separation strategy entirely.

Taken together, these estimates show that the one-point cluster--void asymmetry
does not deliver a competitive $\fnl$ measurement: the primordial reduced skewness
$S_3^{(1)}\sim3\times10^{-4}$ is far below the RSD and bias skewnesses
($\sim0.5$), which cannot be modeled to the required $\sim10^{-4}$ precision.  Unlike scale-dependent halo bias,
$\Delta b(k)\propto\fnl/k^2$, whose large-scale $k^{-2}$ form no astrophysical
process reproduces, the entropy difference carries no such protected signature;
its strengths are statistical, a scalar real-space estimator with many
quasi-linear modes, rather than an immunity to these degeneracies, and it is
best regarded as a complement to, not a competitor of, the
scale-dependent-bias programme.  Its methodological appeal is that the estimator is
a scalar computable ``locally'' in many $(100\,h^{-1}{\rm Mpc})^3$ subsamples of
large, deep surveys like DESI~\cite{DESI2024PNG} and Euclid~\cite{Euclid2024},
without the cosmic-variance limitation that afflicts the few low-wavenumber modes
carrying the scale-dependent-bias signal; but this advantage does not overcome the
fundamental smallness of $S_3^{(1)}\sim3\times10^{-4}$, which keeps the primordial
contribution to $\dH_{\rm CV}$ well below the astrophysical skewness budget.
}

\section{Redshift Evolution of the Entropy}
\label{sec:redshift}

The two pictures developed above --- the $\mathbb{Z}_2$-symmetric Gaussian
baseline of Sec.~\ref{sec:doro} and the gravitationally broken,
asymmetric web of Secs.~\ref{sec:breaking}--\ref{sec:entropy} --- are the two
endpoints of a single temporal evolution.  At high redshift the matter
distribution approaches uniformity, $p_i^{(M)}\to1/4$ and $H^{(M)}\to2$ bits, the
maximum-entropy state in which the four morphologies are equiprobable and the
$\mathbb{Z}_2$ symmetry between clusters and voids is exact.  As gravitational
clustering proceeds, matter evacuates voids, enriches filaments, and concentrates
in clusters, shifting the distribution away from uniformity and \emph{decreasing}
the Shannon entropy in configuration space.  The growth of the Cosmic Web is thus
an entropy-decreasing process in configuration space, an instance of
gravitational self-organization far from thermodynamic
equilibrium~\cite{Peebles1980} (see the Note Added for the complementary
phase-space perspective).  This section makes the evolution quantitative and
shows that its \emph{rate} is fixed by the linear growth rate $f(z)$.

Concurrently with the loss of configuration-space entropy, the information
dimension of each component evolves: at high $z$ the matter field is nearly
Gaussian and monofractal ($D_0\approx D_1\approx D_2\approx3$), while at low $z$
the multifractal spectrum broadens, with $\Delta\alpha=\alpha_\mathrm{max}-
\alpha_\mathrm{min}$ growing as voids become emptier and clusters denser.  The
\emph{multifractal entropy rate}
\begin{equation}
  \dot H_\mathrm{mf}(z) \equiv \frac{d}{dz}\!\left[D_1(z)\,\log (1/\ell)\right]
  \label{eq:Hdot}
\end{equation}
provides a redshift-dependent measure of the rate at which structural information
is generated by gravitational collapse.  We derive this rate analytically below
in terms of the linear growth factor $D(z)$ and the growth rate $f(a)$.

\subsection{Growth factor and tidal eigenvalue variance}
\label{ssec:growth}

The linear density variance smoothed on scale $R$ evolves as
\begin{equation}
  \sigma_\delta^2(R,z) = D(z)^2\,\sigma_\delta^2(R,0),
  \label{eq:sigmaz}
\end{equation}
where $D(z)\equiv D(a)/D(a_0)$ is the linear growth factor normalized to unity at
$z=0$, satisfying
\begin{equation}
  \ddot D + 2H(z)\dot D = \frac{3}{2}\Omega_{m,0} H_0^2 (1+z)^3 D,
  \label{eq:growth_ode}
\end{equation}
with $H(z)=H_0 E(z)$, $E^2(z)=\Omega_{m,0}(1+z)^3+\Omega_\Lambda$.  Since
$\sigma_\lambda^2=\sigma_\delta^2/15$, see Eq.~\eqref{eq:DoroshkevichI}, the
eigenvalue variance scales identically:
\begin{equation}
  \sigma_\lambda(z) = \sigma_\lambda(0)\,D(z).
  \label{eq:sigmalambdaz}
\end{equation}
The \emph{growth rate} is the logarithmic derivative of $D$ with respect to the
scale factor $a=(1+z)^{-1}$,
\begin{equation}
  f(a) \equiv \frac{d\ln D}{d\ln a},
  \label{eq:growth_rate}
\end{equation}
which, in the Linder~\cite{Linder2005} approximation, is well fitted by
\begin{equation}
  f(z) \approx \Omega_m(z)^\gamma, \qquad \gamma\simeq 0.55,
  \label{eq:linder}
\end{equation}
with $\Omega_m(z)=\Omega_{m,0}(1+z)^3/E^2(z)$.  In a flat $\Lambda$CDM universe
with $\Omega_{m,0}=0.3$, $f(0)\approx0.52$ and $f\to1$ as $z\to\infty$ (the
matter-dominated Einstein--de~Sitter limit).

\subsection{Redshift evolution of the differential entropy}
\label{ssec:hz}

Substituting Eq.~\eqref{eq:sigmalambdaz} into Eq.~\eqref{eq:hlambda_final}, the
tidal eigenvalue entropy acquires an explicit redshift dependence:
\begin{equation}
  h[\bm{I}](z) = h[\bm{I}](0) + 3\ln D(z),
  \label{eq:hz}
\end{equation}
where $h[\bm I](0)=\log(8\pi\sqrt5\,\sigma_\lambda^6(0)/e)+\tfrac32(\gamma-
\log(3\sigma_\lambda^2(0)/2))$ is the present-day value.  This compact result
follows directly from the $\sigma_\lambda^6/\sigma_\lambda^3=\sigma_\lambda^3$
dependence in Eq.~\eqref{eq:hlambda_final}.  The analogous expressions for the
density and full-tensor entropies are
\begin{align}
  h[\delta](z) &= h[\delta](0) + \ln D(z),
  \label{eq:hdeltaz}\\[2mm]
  h[\mathbf{T}](z) &= h[\mathbf{T}](0) + 5\ln D(z), \label{eq:hTz}
\end{align}
confirming that the $5{:}1$ ratio $h[\mathbf{T}]\simeq5\,h[\delta]$ of
Eq.~\eqref{eq:hT} is \emph{preserved at all redshifts} in linear theory, since
both sides shift by the same number of $\ln D$ factors up to an integer multiple.
As $z\to\infty$, $D(z)\to0$ and all three differential entropies diverge to
$-\infty$, reflecting the approach to perfect Gaussianity and spatial uniformity
(maximum discrete entropy $H^{(M)}\to2$ bits, minimum differential entropy).

Differentiating Eq.~\eqref{eq:hz} and using $d\ln D/dz=-f(z)/(1+z)$ from
Eq.~\eqref{eq:growth_rate},
\begin{equation}
  \frac{dh[\bm I]}{dz} = -\frac{3f(z)}{1+z} = 3\,\frac{d\ln D}{dz},
  \label{eq:dhdz}
\end{equation}
{which is negative: the \emph{differential} entropy $h[\bm I]$ decreases
toward the past, because the field variance $\sigma_\delta\propto D(z)$ shrinks
(consistent with $h\to-\infty$ as $z\to\infty$).  This is the opposite trend to
the \emph{discrete} morphological entropy $H^{(M)}$, which rises toward its
uniform $2$-bit maximum at high $z$; the two measures have opposite redshift
dependence and must not be conflated.  In the Einstein--de~Sitter limit
($f\to1$, $D\propto(1+z)^{-1}$) the rate simplifies to $dh/dz=-3/(1+z)$,
integrating to $h(z)-h(0)=-3\ln(1+z)$.}

\subsection{Master relation: entropy rate and growth rate}
\label{ssec:master}

Substituting Eq.~\eqref{eq:sigmalambdaz} into the definition
Eq.~\eqref{eq:Hdot} and using $D_1(z)\propto h[\bm I](z)/\log(1/\ell)$ from
Eq.~\eqref{eq:D1} gives the central result of this section,
\begin{equation}
  \dot{H}_\mathrm{mf}(z)
    = \frac{d}{dz}\bigl[D_1(z)\log(1/\ell)\bigr]
    = -\frac{3f(z)}{1+z}
    = 3\,\frac{d\ln D}{dz}.
  \label{eq:master}
\end{equation}
This relation has three immediate consequences.

\textit{Scale independence.}  The factor $\log(1/\ell)$ cancels exactly, so
$\dot H_\mathrm{mf}$ is independent of the smoothing scale $\ell$: the entropy
rate is a purely cosmological quantity, carrying no dependence on the resolution
at which the field is probed.

{This scale-invariance is at first sight surprising, since the T-web
classification is itself intrinsically scale-dependent.  The resolution is to
distinguish the \emph{value} of an entropy from its \emph{rate of change}.  In
linear theory the growth is scale-independent, $\delta(\mathbf{x},z)=D(z)\,
\delta(\mathbf{x},0)$, with every Fourier mode amplified by the same $D(z)$;
hence the variance factorises at any fixed scale,
$\sigma_\delta^2(R,z)=D^2(z)\,\sigma_\delta^2(R,0)$, and
$d\ln\sigma_\delta(R,z)/dz=d\ln D/dz$ is independent of $R$.  Because the tidal
entropy depends on the scale only through the variance,
$h[\bm I](R,z)=\text{const}(R)+3\ln\sigma_\delta(R,z)$, the smoothing scale enters
only the additive zero-point and cancels in the derivative: linear growth
rescales the \emph{width} of the tidal-field PDF while leaving its \emph{shape}
fixed.  The classification (the fractions, $H^{(M)}$, $h[\bm I]$, $D_1$) remains
fully scale-dependent in value; only its logarithmic rate of change with redshift
is universal, inherited entirely from the scale-independent $D(z)$.  This exact
invariance is therefore a null prediction of $\Lambda$CDM linear theory: it would
be broken by scale-dependent growth (massive neutrinos, clustering dark energy or
modified gravity with $D(k,z)$) and by nonlinear mode coupling, so that a measured
residual scale-dependence of $\dot H_\mathrm{mf}$ would itself probe departures
from scale-independent growth.}

\textit{Direct measurement of $f(z)$.}  Equation~\eqref{eq:master} inverts to
\begin{equation}
  f(z) = -\frac{(1+z)}{3}\,\dot H_\mathrm{mf}(z),
  \label{eq:f_from_H}
\end{equation}
so the growth rate, the primary discriminator between general relativity and
modified-gravity theories~\cite{Linder2005}, can be read off directly from the
time derivative of the tidal Shannon entropy.  Combined with
Eq.~\eqref{eq:linder}, a measurement of $\dot H_\mathrm{mf}$ at a single redshift
determines $\Omega_m(z)$, and its evolution tests the $\gamma\simeq0.55$
prediction of GR.

\textit{Connection with $D(z)$ and $\sigma_8$.}  Integrating
Eq.~\eqref{eq:master} between two redshifts $z_1<z_2$,
\begin{equation}
  h[\bm I](z_1) - h[\bm I](z_2) = 3\ln\frac{D(z_1)}{D(z_2)},
  \label{eq:Delta_h}
\end{equation}
so the ratio of growth factors is directly exponential in the entropy difference,
\begin{equation}
  \frac{D(z_1)}{D(z_2)} = \exp\!\left[\frac{h[\bm I](z_1)-h[\bm I](z_2)}{3}\right].
  \label{eq:D_from_h}
\end{equation}
Since {$h[\bm I](z)=\mathrm{const}+3\ln\sigma_\delta(z)$}, this reduces to
$D(z_1)/D(z_2)=\sigma_\delta(z_1)/\sigma_\delta(z_2)$, the standard growth-factor
estimator, here derived as a consequence of the information-theoretic
construction.  At fixed smoothing scale $R=8\,h^{-1}$Mpc,
$\sigma_\delta(z)=\sigma_8 D(z)$, so $h[\bm I](0)$ is directly sensitive to
$\sigma_8$.  {The sensitivity is a Fisher-type derivative,
$\partial h[\bm I]/\partial\ln\sigma_\delta=3$, so that an \emph{absolute} entropy
error $\Delta h$ (in nats) maps to
$\Delta(\sigma_8 D)/(\sigma_8 D)=\Delta h/3$; a measurement of $h$ to
$\Delta h=0.015$ nats would yield a $0.5\%$ constraint on $\sigma_8 D(z)$.  This
figure is, however, an idealized best case rather than a forecast.  The
observable is the differential entropy of a reconstructed tidal-eigenvalue field:
one reconstructs $\delta$ from a survey (with bias and redshift-space modelling),
solves the Poisson equation for $\Psi_{ij}$ on a grid smoothed at $R$,
diagonalizes, and estimates $h$ from the eigenvalue samples with a
$k$-nearest-neighbour or kernel estimator.  Three caveats temper the precision:
(i) in the Gaussian limit $h$ is a monotonic function of $\sigma_\delta$ alone, so
measuring $h$ is equivalent to measuring the variance $\sigma_8 D$ and adds no new
information in that regime, its value lies in the non-Gaussian corrections that
break this degeneracy; (ii) differential-entropy estimators are biased and
converge slowly in several dimensions, so reaching $\Delta h\sim10^{-2}$ nats is
demanding; and (iii) the analytic Doroshkevich form is a linear, real-space
result, while the measured entropy carries nonlinear, redshift-space, shot-noise
and mask contributions that must be forward-modelled.  The $0.5\%$ figure is thus
best read as motivation for a dedicated Fisher analysis on mock catalogues: the
statistic is in principle sensitive to $\sigma_8 D$ at this level, but realizing
it requires controlling these systematics.}

\subsection{Information dimension and the growth rate}
\label{ssec:D1z}

Equation~\eqref{eq:master} also constrains the evolution of the information
dimension $D_1(z)$ of Eq.~\eqref{eq:D1}.  From $H_\mathrm{mf}=D_1(z)\log(1/\ell)$,
\begin{equation}
  \frac{dD_1}{dz} = -\frac{3f(z)}{(1+z)\log(1/\ell)},
  \label{eq:dD1dz}
\end{equation}
which integrates to
\begin{equation}
  D_1(z) = D_1(0) + \frac{3\ln D(z)}{\log(1/\ell)}.
  \label{eq:D1z}
\end{equation}
Since $\ln D(z)<0$ for all $z>0$, the information dimension \emph{decreases} from
its high-$z$ value $D_1\approx3$ (near-Gaussian monofractal) toward its
present-day value $D_1(0)<3$.  The deficit $D_0-D_1(z)$, which quantifies the
degree of multifractality (Sec.~\ref{sec:fractal}), therefore grows as
\begin{equation}
  \frac{d}{dz}\bigl[D_0 - D_1(z)\bigr] = \frac{3f(z)}{(1+z)\log(1/\ell)} > 0,
  \label{eq:multifrac_rate}
\end{equation}
confirming that the multifractal character of the mass distribution increases
monotonically as structure forms, faster in epochs of rapid growth (larger $f$)
and slower in dark-energy-dominated epochs.  Figure~\ref{fig:DFZ} shows the
redshift dependence of these quantities for flat $\Lambda$CDM.  The rate
$\dot H_\mathrm{mf}$ is largest in magnitude at low redshift where structure
growth is suppressed by dark energy; the matter-to-$\Lambda$ transition at
$z\approx0.3$ marks an inflection point where $\ddot H_\mathrm{mf}$ changes sign.
At high redshift $\dot H_\mathrm{mf}\to-3/(1+z)$ (EdS limit), and the entropy
evolves logarithmically, $h(z)\to h(0)-3\ln(1+z)$.  The combination
$|\dot H_\mathrm{mf}(z)(1+z)|=3f(z)$ provides a direct, epoch-by-epoch measure of
the growth rate.

\begin{figure}
    \centering
    \includegraphics[width=0.99\linewidth]{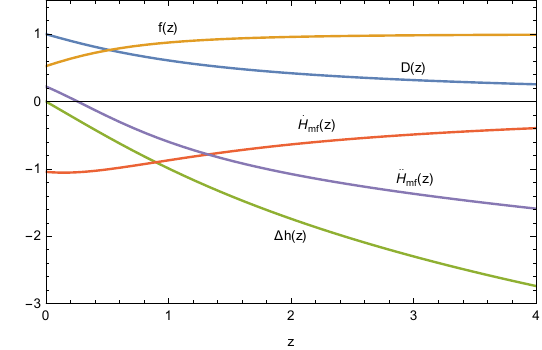}
    \caption{Redshift evolution of the growth factor $D(z)$, growth rate $f(z)$,
    differential entropy shift $\Delta h\equiv h(z)-h(0)$ (in nats), and the
    multifractal entropy rate $\dot H_\mathrm{mf}(z)$ (in nats per unit redshift),
    for flat $\Lambda$CDM with $\Omega_{m,0}=0.3$ and $\sigma_8=0.8$.  At $z=0$:
    $h(0)=\log(8\pi\sqrt5\,\sigma_\lambda^6/e)+\tfrac32(\gamma-\log(3\sigma_\lambda^2/2))$.}
    \label{fig:DFZ}
\end{figure}

\subsection{Critical redshifts and observational windows}
\label{ssec:zcrit}

The successive approximations underlying the symmetry-breaking analysis ---
Edgeworth expansion, Gaussian density floor, single-stream Zel'dovich map --- each
have a regime of validity set by the local amplitude
$\sigma(R,z)=\sigma_0(R)\,D(z)$.  As the field grows it passes through four
distinct regimes, each marked by a critical variance $\sigma_\mathrm{crit}$;
identifying them is essential for interpreting any observed $\dH_{\rm CV}$ or
$\Rcv$ and for choosing the optimal redshift window for a PNG measurement.

\textit{Linear-to-skewed transition} ($\sigma_\mathrm{crit}\approx0.10$).  The
Edgeworth parameter $\eps=\Stot\sigma$ controls the magnitude of all
symmetry-breaking effects.  For gravity alone, $\Stot(R)\approx3.36$ at
$R=8\,h^{-1}$Mpc, so $\eps\approx0.336$ at $\sigma=0.10$, giving asymmetry ratios
$\Rcv-1\approx0.205$ and $\dH_{\rm CV}\approx0.036$ bits, both $\sim20\%$ of
their $z=0$ values, and a fractional shift in $\FC$ of
$A\eps/F_{\rm G}\approx3\%$.  \emph{Above} this redshift ($\sigma<0.10$) the
gravitational breaking becomes small, but the primordial contribution is smaller
still: for $|\fnl|=1$ at $R=8\,h^{-1}$Mpc it gives
$\eps^\mathrm{PNG}\approx3\times10^{-4}\,\sigma$, far below the gravitational
baseline at every accessible redshift.  There is thus no regime in which a
realistic primordial skewness dominates the one-point cluster--void asymmetry.

\textit{Hard lower bound} ($\sigma_\mathrm{crit}\approx0.27$).  The Gaussian PDF
assigns probability $\Phi(-1/\sigma)$ to regions below the hard floor
$\delta=-1$.  At $\sigma=0.27$, $\Phi(-3.7)\approx10^{-4}$; at $\sigma=0.32$,
$\Phi(-3.1)\approx10^{-3}$.  Once this probability is non-negligible the Gaussian
PDF has support on forbidden configurations and the Log-Doroshkevich PDF of
Sec.~\ref{sec:logdoro} is required.  The transition is smooth: the lognormal
correction first appears in the skewness of the density-weighted fractions and
becomes dominant in the cluster-void density contrast (reaching
$\langle1+\delta\rangle_C/\langle1+\delta\rangle_V\approx4.8$ at
$\sigma_\delta=0.5$).  At $R=8\,h^{-1}$Mpc this corresponds to $z\approx3$.

\textit{Moderate nonlinearity} ($\sigma_\mathrm{crit}\approx0.60$).  Here
$\eps\approx2.0$ for gravity alone is no longer small and the Edgeworth expansion
has broken down parametrically: the density PDF develops substantial kurtosis
$S_4\sigma^2\sim O(1)$, and the Zel'dovich matter fractions have grown well above
their volume counterparts (Table~\ref{tab:matter}).  Below this redshift the
Edgeworth formulae for $\Rcv$ and $\dH_{\rm CV}$ overestimate the true breaking
and must be replaced by a numerical treatment using the Log-Doroshkevich or
$N$-body distributions.

\textit{Shell crossing} ($\sigma_\mathrm{crit}\approx1$).  When $\sigma\sim1$ the
typical largest eigenvalue satisfies $D\lambda_1\sim D\sigma/\sqrt5\sim0.45\,D$,
and first shell crossings ($D\lambda_1=1$) occur near
$\sigma\sim\sqrt5\,D\sigma_\mathrm{DS}\sim1$.  The Zel'dovich map then becomes
multivalued, the single-stream density
$1/[(1-D\lambda_1)(1-D\lambda_2)(1-D\lambda_3)]$ diverges, and multistreaming
thermalizes the tidal information, erasing the one-to-one correspondence between
initial eigenvalue signs and final web type.  PNG information encoded in those
signs is largely lost on scales where $\sigma>1$.

\textit{Observational windows.}  Table~\ref{tab:zcrit} maps these critical
variances to redshifts for three smoothing scales.  The most favourable window
for a T-web PNG measurement is the quasi-linear regime between the
linear-to-skewed and hard-wall thresholds,
\begin{equation}
  0.10 \lesssim \sigma(R,z) \lesssim 0.27.
  \label{eq:window}
\end{equation}
For $R=8\,h^{-1}$Mpc this is $3\lesssim z\lesssim9$; for $R=2\,h^{-1}$Mpc,
$8\lesssim z\lesssim23$; for $R=0.5\,h^{-1}$Mpc, $19\lesssim z\lesssim54$.
Spectroscopic surveys (DESI, Euclid) reach $z\lesssim2$ and access mainly the
mildly nonlinear regime $\sigma\gtrsim0.3$ at cluster scales, where the Edgeworth
treatment must be supplemented by the Log-Doroshkevich PDF.  Future 21-cm
intensity mapping~\cite{Contarini2023} and high-$z$ programmes targeting $z>5$
would reach the quasi-linear window at $R=8\,h^{-1}$Mpc, where the perturbative
prediction Eq.~\eqref{eq:DHcv} is precise to better than $\eps^2\sim10\%$.

\begin{table}[h]
  \caption{Critical redshifts $z_\mathrm{crit}$ defined by
    $\sigma(R,z_\mathrm{crit})=\sigma_\mathrm{crit}$, for flat $\Lambda$CDM with
    $\Omega_m=0.3$, $\sigma_8=0.812$.  ``$<0$'' means the scale has
    $\sigma_0<\sigma_\mathrm{crit}$ and never reaches this regime.}
  \label{tab:zcrit}
  \begin{ruledtabular}
  \begin{tabular}{lccc}
    Mechanism ($\sigma_\mathrm{crit}$) & $R=0.5$ & $R=2$ & $R=8$ \\
    & $h^{-1}$Mpc & $h^{-1}$Mpc & $h^{-1}$Mpc \\
    \hline
    Linear skewness (0.10) & $z=54$ & $z=23$ & $z=9$ \\
    Hard wall (0.27)       & $z=19$ & $z=8$  & $z=3$ \\
    Moderate NL (0.60)     & $z=8$  & $z=3$  & $z=0.6$ \\
    Shell crossing (1.00)  & $z=5$  & $z=1$  & $<0$ \\
  \end{tabular}
  \end{ruledtabular}
\end{table}

\section{Discussion}
\label{sec:discussion}

The picture assembled here rests on a small number of analytic facts about the
Doroshkevich distribution and their information-theoretic consequences.  Three
distinct information measures characterise the Cosmic Web at different levels of
description: the discrete Shannon entropy of the T-web morphological
classification, the continuous differential entropy of the tidal eigenvalue
distribution, and the multifractal information dimension.  Each probes a
complementary aspect of the total budget, and together they paint a richer
statistical picture of large-scale structure than the density contrast alone.

\textit{The $\mathbb{Z}_2$ baseline is not a coincidence.}  That the physical
$6\times6$ tidal-tensor Monte Carlo and the direct Doroshkevich PDF integral both
return $\FVC=7.96\%$, matching the Shandarin--Doroshkevich--Zel'dovich~\cite{Shandarin1983}
value of $\approx8\%$, follows from the isotropic covariance
$\langle\Psi_{ij}^2\rangle/\langle\Psi_{ii}^2\rangle=1/3$, which renders the
quadratic form $Q=(u^2+5v)/(2\sigma_\delta^2)$ rotation-invariant and lets the
Vandermonde factor cancel in the $(u,v,t)$ parametrization of
Appendix~\ref{app:fractions}.  The exact equality $\FVC=\FVV$ is then a
$\mathbb{Z}_2$ symmetry of the Gaussian theory, and it furnishes the
maximum-entropy reference state $H^{(V)}_{\rm G}=1.632$ bits from which all
subsequent symmetry breaking is measured.

\textit{Tidal entropy as a cosmological probe.}  The differential entropy
$h[\bm\lambda]$ of Eq.~\eqref{eq:hlambda_final} depends on
$\sigma_\lambda\propto\sigma_\delta$, hence on the matter power spectrum $P(k)$
and the growth factor $D(z)$.  Measuring it from reconstructed tidal fields gives
a single-number summary statistic sensitive to the shape of $P(k)$ and the growth
of structure, complementing standard two-point statistics~\cite{Leclercq2015,Jasche2013},
the Fisher-information approach of Bonnaire et al.~\cite{Bonnaire2022} (whose
environment-split power spectra improve neutrino-mass constraints by a factor
$\sim15$), and the classifier-comparison framework of
Leclercq et al.~\cite{Leclercq2016}, which used the Jensen--Shannon divergence to
rank T-web, DIVA and ORIGAMI~\cite{Falck:2012ai} classifiers.

\textit{Shear information beyond density.}  The identity
$h[\mathbf{T}]=5\,h[\delta]$ at linear order {(at equal variance; the literal
ratio is scale-dependent, Sec.~\ref{ssec:hshear})} reflects that the shear
carries five of the six tensor degrees of freedom, so that a {\emph{local}
(one-point)} density statistic discards the bulk of the linear morphological
information.  {This is not discarded by a field-level density analysis: the
tidal tensor is fixed by $\delta(\mathbf{x})$ through the Poisson equation, and in
the linear Gaussian regime carries no two-point information beyond $P(k)$.}
Field-level inference methods that reconstruct the full
$\Pij$~\cite{Jasche2013,Porqueres2021} use it; current summary statistics
based on the {local} density alone do not.  The multifractal discriminators $D_q$ and the
singularity spectrum $g(\alpha)$ add a further, model-sensitive channel:
Gaite~\cite{Gaite2019} showed that a nonlacunar multifractal geometry with
$\alpha_\mathrm{min}\approx1$ describes web structure, hierarchical clustering and
halo distributions simultaneously, and modified-gravity theories that alter
$f\sigma_8$ also alter the width $\Delta\alpha$.

{\textit{Entropy budget versus parameter constraints.}  The entropy measures
used here --- the discrete $H^{(M)}$, the differential $h[\bm I]$, and the
multifractal $D_1$ --- are properties of the field's distribution \emph{at fixed
cosmology}; they quantify descriptive and morphological information, which is
conceptually distinct from the Fisher information
$F_{ab}=-\langle\partial_a\partial_b\ln\mathcal{L}\rangle$ that sets parameter
error bars through the Cram\'er--Rao bound.  A statistic can carry large entropy
yet little sensitivity to $\theta=\{\Omega_m,\sigma_8,f,\dots\}$, or vice versa.
Thus the statement that filaments dominate the entropy budget is descriptive and
does not by itself imply that filaments dominate the cosmological constraints; for
the latter we refer to the Fisher analyses of Bonnaire et al.~\cite{Bonnaire2022}
and Leclercq et al.~\cite{Leclercq2016}.  The one place where the present
construction legitimately yields a constraint is the sensitivity relation
$\partial h[\bm I]/\partial\ln\sigma_\delta=3$ of Sec.~\ref{ssec:master}, a
Fisher-type derivative valid in the Gaussian, linear regime.}

{
\textit{The Log-Doroshkevich extension.}  The lognormal
approximation~\cite{Coles1991,Xavier2016,Carron2012} has long served as a cheap
surrogate for the nonlinear density PDF, and the Log-Doroshkevich distribution
carries it to the eigenvalue level.  It enforces the hard floor $\delta>-1$
automatically and develops a power-law overdense tail
$P(\lambda_1)\propto\lambda_1^{-5/\sigma_g^2}$.  That the $\mathbb{Z}_2$ equality
$F_{\rm C}^{\rm LN}=F_{\rm V}^{\rm LN}$ survives this strongly nonlinear map is
worth stressing: the eigenvalue sign structure is untouched, so all symmetry
breaking is confined to the density-weighted sector.

\textit{Primordial non-Gaussianity.}  Scale-dependent
bias~\cite{Dalal2008,Matarrese2008,Desjacques2018},
$\Delta b(k)\propto\fnl/k^2$, remains the most competitive single-observable
handle on $\fnl$, with recent analyses~\cite{DAmico2022,Cabass2022} reaching
$\sigma(\fnl)\sim20$.  The entropy difference studied here works instead on scales
$R\sim2$--$8\,h^{-1}$Mpc, where the field is still perturbative and far more
independent modes are available; it is a complement to the bias method rather than
a substitute.  Because $\dH_{\rm CV}$ is a real-space scalar built from one-point
web statistics, it can be estimated locally within many independent
$(100\,h^{-1}{\rm Mpc})^3$ subvolumes, averaging down the cosmic variance that
limits scale-dependent bias.  We stress, however, that the primordial reduced
skewness is intrinsically small, $S_3^{(1)}\sim3\times10^{-4}$, so the one-point
entropy asymmetry is dominated by the gravitational skewness and by astrophysical
(RSD and bias) contributions; its value for $\fnl$ is as a consistency check
complementing scale-dependent bias, not as a competitive stand-alone estimator.
Both are within reach of three-dimensional reconstructions from
DESI~\cite{DESI2024PNG} and Euclid~\cite{Euclid2024}.
}

\textit{Continuous entropy, growth, and prior work.}  The same linear scaling
$\dH_{\rm CV}\propto\sigma_0(R)D(z)$ and the master relation
$\dot H_\mathrm{mf}=-3f(z)/(1+z)$ make the tidal entropy a growth probe: a ratio
of entropy differences at two redshifts returns the growth factor directly,
$D(z_1)/D(z_2)=|\dH(z_1)|/|\dH(z_2)|$, independent of galaxy bias.  This
complements the Kullback--Leibler entropy of Hosoya et al.~\cite{Hosoya2004}, the
Shannon-entropy homogeneity tests of Pandey~\cite{Pandey2013}, and the
algorithmic-complexity bounds of Vazza~\cite{Vazza2017,Vazza2020}.  Finally, the
entropy difference has a natural place in the phase-space picture of
Kitaura \& Sinigaglia~\cite{Kitaura:2026ike}: with
$S_\mathrm{ps}=S_x+S_{v|x}$ rising monotonically, the configuration-space part
$S_x\equiv H^{(M)}$ falls as structure forms while $S_{v|x}$ makes up the
difference, and $\dH_{\rm CV}$ isolates the internal asymmetry that grows within
$S_x$ as the $\mathbb{Z}_2$ symmetry breaks.

\section{Conclusions}
\label{sec:conclusions}

We have presented an analytic, information-theoretic treatment of the Cosmic Web
that unifies two complementary facts about the tidal deformation tensor: that its
Gaussian eigenvalue statistics possess an exact $\mathbb{Z}_2$ symmetry, and that
the tensor carries far more cosmological information than the density contrast it
traces.  Our main results are the following.

{
\textit{The $\mathbb{Z}_2$ symmetry and the Gaussian baseline.}  The Doroshkevich
PDF is invariant under $\delta\to-\delta$, an exact $\mathbb{Z}_2$ symmetry that
equates the cluster and void volume fractions, and likewise the filament and wall
fractions, at every scale and redshift for which the field stays Gaussian.
Sampling the physical $6\times6$ covariance and integrating the PDF both yield
$\FVC=\FVV=7.96\%$ and $\FVF=\FVW=42.04\%$, giving the maximum-entropy
classification value $H^{(V)}_{\rm G}=1.632$ bits.
}

\textit{Tidal eigenvalue entropy.}  The Doroshkevich PDF yields a closed-form
differential entropy for $(\lambda_1,\lambda_2,\lambda_3)$,
Eq.~\eqref{eq:hlambda_final}, splitting into a Gaussian part controlled by
$\sigma_\delta$ and an anisotropy term encoding the eigenvalue spread.  The
traceless shear tensor carries five degrees of freedom, so
{$h[\mathbf{T}]=5\,h[\delta]$ at equal variance (a scale-dependent ratio,
Sec.~\ref{ssec:hshear}), and a \emph{local} (one-point) density statistic retains
only one of the six tensor components, information recovered by field-level,
but not by local-density, analysis}; the invariants $\mathcal{Q}$ and $\mathcal{A}$ are statistically independent of $\delta$.

\textit{Entropy budget and multifractality.}  The discrete matter-weighted
entropy is $H^{(M)}\approx1.70$ bits, dominated by filaments and walls near the
$-p\log p$ maximum at $p^\star=e^{-1}$; the low Hausdorff dimension of filaments
($D_H\approx1.8$) gives them the highest information density of any component.
Each component carries a spectrum of R\'enyi dimensions $D_q$, with the deficit
$D_0-D_1$ quantifying multifractality and growing monotonically as structure
forms.

{
\textit{The Log-Doroshkevich distribution.}  The lognormal map
$\lambda_i=e^{\nu_i}-1$ carries the Gaussian PDF into the Log-Doroshkevich
distribution of Eq.~\eqref{eq:logDoro}, enforcing the floor $\delta>-1$ and
developing a power-law overdense tail, while preserving $\FC^{\rm LN}=\FV^{\rm LN}$
exactly: the $\mathbb{Z}_2$ symmetry survives the full nonlinear map at the level
of volumes and is broken only by density weighting.  Already at
$\sigma_\delta=0.5$ the mean cluster density exceeds the void density by a factor
$4.8$.

\textit{Symmetry breaking and the entropy decrease.}  Gravitational evolution
breaks $\mathbb{Z}_2$ through the skewness of the density field.  The
cluster-void and filament-wall asymmetries follow $\Rcv-1=0.6088\,\eps$ and
$\Rfw-1=-0.3164\,\eps$, with $\eps=\Stot(R)\,\sigma_0(R)\,D(z)$ and the numerical
coefficients fixed by Gaussian-quantile geometry alone.  The cluster-void Shannon
entropy difference is $\dH_{\rm CV}=0.1070\,\eps$ bits.  Through the Zel'dovich
matter fractions the classification entropy falls from the Gaussian $1.63$ bits to
$\approx1.42$ bits by the onset of shell crossing, the residual information
carried by the near-degenerate cluster--filament pair.

\textit{Primordial non-Gaussianity.}  A primordial skewness
$\Spng(R)=S_3^{(1)}(R)\,\fnl$, with $S_3^{(1)}(R=8\,h^{-1}{\rm Mpc})\simeq3\times10^{-4}$,
adds to the gravitational term $\Sgrav=34/7+\neff$ and shifts the entropy
difference by $|\dH_\mathrm{PNG}|\simeq3.2\times10^{-5}\,|\fnl|\,\sigma_0(R)\,D(z)$
bits.  Because this is a factor $\sim10^{-4}|\fnl|$ below the gravitational
baseline, the one-point cluster--void asymmetry is dominated by gravity, and the
tidal-web entropy provides only a weak, complementary check on $\fnl$; the leading
large-scale-structure signature of local PNG remains the scale-dependent halo bias
$\Delta b(k)\propto\fnl/k^2$.
}

\textit{Redshift evolution and the growth rate.}  The multifractal entropy rate
obeys $\dot H_\mathrm{mf}(z)=-3f(z)/(1+z)$, equating the information-theoretic
entropy rate to the linear growth rate independently of smoothing scale, so that
$\dot H_\mathrm{mf}$ measured across epochs constrains $f(z)$ and tests deviations
from general relativity.  The growth-factor ratio
$D(z_1)/D(z_2)=\exp[(h(z_1)-h(z_2))/3]$ makes the tidal Shannon entropy at
$R=8\,h^{-1}$Mpc {in principle sensitive to $\sigma_8 D(z)$ at the $0.5\%$
level for a $0.015$-nat entropy measurement, an idealized best case whose
realization requires controlling reconstruction and estimator systematics
(Sec.~\ref{ssec:master})}.

The principal observational implication is that field-level reconstructions
inferring $\delta$ \emph{and} the full tidal tensor $\Pij$ from
DESI~\cite{DESI2016}, Euclid~\cite{Laureijs2011} and the Vera C.\ Rubin
Observatory~\cite{Ivezic2019} would recover the {morphological information
discarded by \emph{local}-density-only analyses (five of the six tensor degrees
of freedom)}, while the symmetry-breaking
observables $\Rcv$ and $\dH_{\rm CV}$ provide a real-space, locally estimable
route to both the growth rate and primordial non-Gaussianity.

\section*{Note Added}

After completion of this work, Kitaura~\cite{Kitaura2026} posted a paper entitled
\textit{Emergence of Complex Structures} that addresses a closely related set of
questions from a substantially different angle.  Both works place the tidal
deformation tensor and its eigenvalue spectrum at the centre of the analysis,
regard the Gaussian random field as the maximum-entropy baseline from which the
Cosmic Web departs under gravitational evolution, and argue that the density
contrast $\delta=\Tr\bm\Psi$ is an incomplete characterization of the available
cosmological information.  Both also grapple with the same entropy paradox, that
structure formation creates geometric order in an evolving Universe, and arrive
at consistent resolutions.  Kitaura resolves it through the phase-space
decomposition $S_\mathrm{ps}=S_x+S_{v|x}$, where the full phase-space entropy
increases monotonically while the projected spatial entropy
$S_x\equiv H^{(M)}$ decreases as filaments and clusters form; shell crossing
activates velocity degrees of freedom ($S_{v|x}>0$) that absorb the information
shed by the configuration-space field.  The present paper reaches the same
conclusion analytically, showing that $H^{(M)}(z)$ decreases as clustering
proceeds.

Kitaura~\cite{Kitaura2026} also shows, through a Lagrangian--Eulerian transport
analysis, that long-range tidal interaction becomes relevant already at moderate
overdensity, the dynamical counterpart of the result derived here that the
traceless shear tensor carries five times the differential entropy of $\delta$:
both statements quantify how much morphological information lies beyond the density
field, the one through transport geometry and the other through the Doroshkevich
entropy calculus.

The two papers address different regimes and provide different outputs.  Kitaura
treats the full six-dimensional phase-space dynamics, including shell crossing and
multistreaming, and introduces a Landau--Ginzburg free-energy description in which
anisotropy acts as an order parameter driving self-organization; this supplies the
physical mechanism underlying the redshift evolution of $H^{(M)}$ and the
broadening of $g(\alpha)$ described here.  The present paper supplies the analytic
entropy calculus in closed form --- the Doroshkevich differential entropy, the
T-web entropy budget, the factor-of-five shear information advantage, the
multifractal $D_q$ spectrum, the $\mathbb{Z}_2$ symmetry of the Gaussian
classification and its gravitational breaking, and the master relation
$\dot H_\mathrm{mf}=-3f(z)/(1+z)$ --- together with a discussion of observational
classifiers not covered in~\cite{Kitaura2026}.  Read together, the two works cover
the information content of the Cosmic Web from the linear Gaussian regime through
the fully nonlinear, multistreaming regime.

\begin{acknowledgments}
The author thanks David Alonso and Mikel Mart\'in Barandiaran for stimulating
discussions on information theory and large-scale structure, Francisco Kitaura for
discussions on the phase-space entropy, and Istvan Szapudi for insights on the
relevance of voids in the Cosmic Web.  This work was supported by the Research
Project PID2024-159420NB-C43 [MICINN-FEDER] and the Centro de Excelencia Severo
Ochoa Program CEX2020-001007-S at IFT.
\end{acknowledgments}

\appendix

\section{The Doroshkevich Distribution}
\label{app:doro}

We derive here the joint probability distribution of the two leading symmetric
invariants $I_1$ and $I_2$ of the tidal deformation tensor for a homogeneous,
isotropic Gaussian random field, the \emph{Doroshkevich
distribution}~\cite{Doroshkevich1970,Doroshkevich1978}, and compute the
moments $\langle I_1^2\rangle$, $\langle I_2\rangle$ and $\langle\log v\rangle$
analytically, exploiting a factorization that emerges under a natural change of
variables.  These results underpin the T-web classification used throughout the
main text~\cite{Hahn2007}.

\subsection{Gaussian random field and the tidal tensor}
\label{app:setup}

Let $\delta(\mathbf{x})$ be a homogeneous, isotropic Gaussian density field
smoothed on scale $R_s$.  The tidal tensor (the Hessian of the gravitational
potential through the Poisson equation),
\begin{equation}
  \Pij(\mathbf{x}) = \frac{\partial^2\phi}{\partial x_i\,\partial x_j},
  \label{eq:tidalapp}
\end{equation}
is a real symmetric $3\times3$ matrix with six independent components.  Isotropy
fixes their covariance completely,
\begin{equation}
  \langle \Pij\,\Psi_{kl}\rangle
  = \frac{\sig^2}{15}
    \bigl(\delta_{ik}\delta_{jl}+\delta_{il}\delta_{jk}+\delta_{ij}\delta_{kl}\bigr),
  \label{eq:covariance}
\end{equation}
from which $\langle\Psi_{ii}^2\rangle=\sig^2/5$,
$\langle\Psi_{ii}\Psi_{jj}\rangle_{i\neq j}=\sig^2/15$, and
$\langle\Pij^2\rangle_{i\neq j}=2\sig^2/15$.  The joint PDF of the six components
is therefore
\begin{align}
  &\calP(\Pij) \propto \nonumber \\
  &\exp\!\left[
    -\frac{15}{4\sig^2}
    \left(
      \sum_i \Psi_{ii}^2
      + \half\sum_{i\neq j}\Pij^2
      - \frac{1}{5}\!\left(\sum_i \Psi_{ii}\right)^{\!2}
    \right)
  \right].
  \label{eq:jointPDF}
\end{align}

\subsection{Eigenvalue decomposition and the Jacobian}
\label{app:eigenvalues}

Diagonalizing by an orthogonal transformation with ordered eigenvalues
$\lambda_1\geq\lambda_2\geq\lambda_3$ and three Euler angles
$\boldsymbol{\Omega}$, the measure is~\cite{Mehta2004}
\begin{equation}
  \dif^6 T
  = \prod_{i<j}|\lambda_i-\lambda_j|\cdot
    \dif\lambda_1\,\dif\lambda_2\,\dif\lambda_3\cdot\dif\boldsymbol{\Omega},
  \label{eq:jacobian}
\end{equation}
with $\dif\boldsymbol{\Omega}$ the Haar measure on $SO(3)$, whose integral is
absorbed into the normalization.  The exponent in
Eq.~\eqref{eq:jointPDF} reduces to
$\sum_i\lambda_i^2-\tfrac15(\sum_i\lambda_i)^2$, so the joint PDF of the
eigenvalues is
\begin{eqnarray}
  &&\calP(\lambda_1,\lambda_2,\lambda_3)
  = \calN_\lambda\,
  (\lambda_1-\lambda_2)(\lambda_1-\lambda_3)(\lambda_2-\lambda_3)\times \nonumber \\
  && \hspace{0.6cm} \exp\!\left[
    -\frac{15}{4\sig^2}
    \!\left(\sum_i\lambda_i^2-\frac{1}{5}\!\left(\sum_i\lambda_i\right)^{\!2}\right)
  \right],
  \label{eq:PDF_eig}
\end{eqnarray}
with $\calN_\lambda=15^3/(8\pi\sqrt5\,\sig^6)$.  In terms of the symmetric
invariants
\begin{align}
  I_1 &= \lambda_1+\lambda_2+\lambda_3 = \Tr\bm\Psi, \label{eq:I1a}\\
  I_2 &= \lambda_1\lambda_2+\lambda_1\lambda_3+\lambda_2\lambda_3, \label{eq:I2a}\\
  I_3 &= \lambda_1\lambda_2\lambda_3 = \det\bm\Psi, \label{eq:I3a}
\end{align}
the identity $\sum_i\lambda_i^2=I_1^2-2I_2$ turns the exponent into
$-I_1^2/(2\sig^2)-5(I_1^2-3I_2)/(2\sig^2)$.  The squared Vandermonde factor equals
the discriminant
\begin{equation}
  \Delta = I_1^2 I_2^2 - 4I_2^3 - 4I_1^3 I_3 + 18I_1 I_2 I_3 - 27I_3^2,
  \label{eq:discriminant}
\end{equation}
with $\Delta>0$ for real, distinct eigenvalues.

\subsection{Marginalization over $I_3$}
\label{app:marginal}

Changing variables $(\lambda_1,\lambda_2,\lambda_3)\to(I_1,I_2,I_3)$ has Jacobian
$1/|\Delta(\lambda)|$, which cancels the Vandermonde factor.  At fixed $(I_1,I_2)$
the discriminant is a downward quadratic in $I_3$, and integrating between its two
roots gives
\begin{equation}
  \int_{I_3^-}^{I_3^+}\dif I_3 = \frac{4}{27}(I_1^2-3I_2)^{3/2}.
  \label{eq:rootdiff}
\end{equation}
The marginal distribution in $(I_1,I_2)$ is the Doroshkevich PDF
\begin{equation}
  \calP(I_1,I_2)
  = \calN\,
  \left(I_1^2-3I_2\right)^{3/2}
  \exp\!\left(-\frac{I_1^2}{2\sig^2} - \frac{5(I_1^2-3I_2)}{2\sig^2}\right),
  \label{eq:DoroshkevichApp}
\end{equation}
with support $I_2\leq I_1^2/3$ and $\calN=25\sqrt5/(2\pi\,\sig^6)$.

\subsection{Factorization and moments}
\label{app:moments}

The substitution $u=I_1$, $v=I_1^2-3I_2\geq0$ (so $I_2=(u^2-v)/3$ and
$\dif I_1\,\dif I_2=\tfrac13\dif u\,\dif v$) diagonalizes the exponent and
factorizes the distribution,
\begin{equation}
  \calP(u,v) = \frac{25\sqrt5}{6\pi\,\sig^6}
  \, v^{3/2} \exp\!\left(-\frac{5v}{2\sig^2}-\frac{u^2}{2\sig^2}\right),
  \label{eq:PDF_uv}
\end{equation}
with $u\in(-\infty,\infty)$ and $v\geq0$ statistically independent: $u$ is
Gaussian and $v$ follows a Gamma distribution,
\begin{align}
  \calP_u(u) &= \frac{1}{\sig\sqrt{2\pi}}\exp\!\left(-\frac{u^2}{2\sig^2}\right),
  \label{eq:marginal_u}\\[4pt]
  \calP_v(v) &= \frac{25\sqrt5}{3\sig^5\sqrt{2\pi}}\,
  v^{3/2}\exp\!\left(-\frac{5v}{2\sig^2}\right).
  \label{eq:marginal_v}
\end{align}
Their moments, $\langle u^{2n}\rangle=(2\sig^2)^{n}\Gamma[1/2+n]/\Gamma[1/2]$ and
$\langle v^{n}\rangle=(2\sig^2/5)^{n}\Gamma[5/2+n]/\Gamma[5/2]$, give
{$\langle u^2\rangle=\sig^2$ and $\langle v\rangle=\sig^2$, so}
\begin{equation}
  \langle I_1^2\rangle = \sig^2, \qquad
  \langle I_2\rangle = \frac{1}{3}\langle u^2\rangle-\frac{1}{3}\langle v\rangle
  = {0},
  \label{eq:I2_decomp}
\end{equation}
{the mean quadratic invariant vanishing through the exact cancellation of the
isotropic and shear contributions, $\langle I_2\rangle=\tfrac13\langle\delta^2\rangle
-\tfrac12\langle\mathcal{Q}\rangle=\tfrac{\sig^2}{3}-\tfrac12\cdot\tfrac{2\sig^2}{3}=0$;
the corresponding non-vanishing shear amplitude is
$\langle\mathcal{Q}\rangle=\langle\Tr\mathbf{T}^2\rangle=\tfrac{2}{3}\sig^2$.}  The
logarithmic moment that enters the differential entropy
Eq.~\eqref{eq:hlambda_final} is
\begin{equation}
    \langle\log v\rangle = \frac{8}{3} - \gamma + \log\frac{\sig^2}{10},
  \label{eq:logv}
\end{equation}
with $\gamma=0.5772\dots$ the Euler constant.

\section{Volume and Matter Fractions from the Doroshkevich PDF}
\label{app:fractions}

This appendix derives closed-form integral representations for the cluster volume
fraction $\FVC$ and matter fraction $\fMC$ directly from the factorized
Doroshkevich PDF, Eq.~\eqref{eq:Doro}, following the $(u,v,t)$ change of variables.

\subsection{The three-variable density in $(u,v,t)$ space}

Equation~\eqref{eq:Doro} is the marginal over $I_3$.  At fixed $(u,v)$ the three
eigenvalues are real iff $I_3\in[I_3^-,I_3^+]$.  Substituting $x=\mu+u/3$ into the
characteristic polynomial (with $I_1=u$, $I_2=(u^2-v)/3$) gives the depressed
cubic
\begin{equation}
  \mu^3 - \frac{v}{3}\,\mu - t = 0,
  \qquad t\equiv I_3 - c(u,v),
  \label{eq:depressed}
\end{equation}
with three real roots when $\Delta=4v^3/27-27t^2\ge0$, i.e.
\begin{equation}
  t\in\left[-\frac{\ell(v)}{2},\,\frac{\ell(v)}{2}\right],
  \qquad
  \ell(v) = \frac{4v^{3/2}}{27}.
  \label{eq:ellv}
\end{equation}
The centre $c(u,v)=u^3/27-uv/9$ depends on $u$ but the half-width $\ell(v)/2$ does
not, so marginalizing over $t$ recovers the factor $v^{3/2}$ in
Eq.~\eqref{eq:Doro}.  The full three-variable PDF, flat in $t$, is
\begin{equation}
  \dif W = p_0\,e^{-(u^2+5v)/(2\sigma_\delta^2)}\,\dif u\,\dif v\,\dif t,
  \label{eq:dW3}
\end{equation}
with $p_0=675\sqrt5/(24\pi\sigma_\delta^6)$.  The eigenvalues factorize as
\begin{equation}
  \lambda_i = \mu_i(v,t) + \frac{u}{3},
  \qquad i=1,2,3,
  \label{eq:lam_shift}
\end{equation}
where the ordered roots $\mu_1\ge\mu_2\ge\mu_3$ are
\begin{eqnarray}
  \mu_k &=& \frac{2\sqrt{v}}{3}
    \cos\!\left(\frac{\theta_t}{3}-\frac{2\pi(k-1)}{3}\right),
    \nonumber \\[2mm]
  \theta_t &=& \arccos\!\left(\frac{27t}{2v^{3/2}}\right)\in[0,\pi].
  \label{eq:vieta}
\end{eqnarray}
The $\mu_i$ depend only on $(v,t)$; the entire dependence on the trace
$u=I_1=\delta$ enters as a uniform shift $u/3$ of all three eigenvalues.

\subsection{Volume fraction $\FVC$}

The cluster domain $\lambda_3>0$ is equivalent to $u>-3\mu_3(v,t)$.  Integrating
Eq.~\eqref{eq:dW3}, the inner $u$-integral is a Gaussian tail
$\Phi(3\mu_3/\sigma_\delta)$, so
\begin{equation}
  \FVC = \mathbb{E}_{v,t}\!\left[\Phi\!\left(\frac{3\mu_3(v,t)}{\sigma_\delta}\right)\right],
  \label{eq:FC_exact}
\end{equation}
the expectation taken over $v\sim\Gamma[5/2,2\sigma_\delta^2/5]$ and
$t\,|\,v\sim\mathrm{Uniform}[-\ell(v)/2,\ell(v)/2]$.  With
$\Phi_i\equiv\Phi(3\mu_i/\sigma_\delta)$ the remaining fractions are
\begin{alignat}{2}
  \FVF &= \mathbb{E}[\Phi_2 - \Phi_3], &\quad
  \FVW &= \mathbb{E}[\Phi_1 - \Phi_2], \nonumber \\[2mm]
  \FVV &= \mathbb{E}[1-\Phi_1], &
  \label{eq:all_fracs}
\end{alignat}
summing to unity.

\subsection{Proof of $\FVC=\FVV$ from the integral}

Under $t\to-t$ the depressed cubic maps to $\mu^3-(v/3)\mu+t=0$, whose roots are
$-\mu_{4-k}$, so
\begin{equation}
  \mu_3(v,-t) = -\mu_1(v,t).
  \label{eq:t_reversal}
\end{equation}
Hence $\Phi(3\mu_3(v,-t)/\sigma_\delta)=1-\Phi_1(v,t)$, and since $t$ is integrated
uniformly over the symmetric interval,
\begin{equation}
  \FVC = \mathbb{E}[\Phi_3] = \mathbb{E}[\Phi_3(v,-t)] = \mathbb{E}[1-\Phi_1] = \FVV.
\end{equation}
This is the integral-level proof of Theorem~\ref{thm:Z2G}: the $\mathbb{Z}_2$
symmetry follows from the symmetric $t$-domain combined with the root-reversal
identity Eq.~\eqref{eq:t_reversal}.

\subsection{Matter fraction $\fMC$}
\label{app:matter}

In the Zel'dovich approximation with growth factor $D\equiv D(z)$, the density at
a point with eigenvalues $\lambda_i=\mu_i+u/3$ is
$(1+\delta)=\prod_i[1-D(\mu_i+u/3)]^{-1}$, and shell crossing occurs at
$u=u_\star\equiv3(D^{-1}-\mu_1)$.  The cluster matter fraction is
\begin{equation}
  \fMC = \frac{\displaystyle
    \int_0^\infty\!dv\int_{-\ell/2}^{\ell/2}\!dt\int_{-3\mu_3}^{u_\star}
      \frac{e^{-(u^2+5v)/(2\sigma_\delta^2)}}{\prod_i[1-D(\mu_i+u/3)]}\,\dif u}
  {\displaystyle
    \int_0^\infty\!dv\int_{-\ell/2}^{\ell/2}\!dt\int_{-\infty}^{u_\star}
      \frac{e^{-(u^2+5v)/(2\sigma_\delta^2)}}{\prod_i[1-D(\mu_i+u/3)]}\,\dif u}.
  \label{eq:fCM_full}
\end{equation}
The $u$-dependence of the denominator prevents an analytic reduction for finite
$D$; it is evaluated by Monte Carlo over $(v,t)$ with the inner $u$-integral done
numerically.

\textit{Linear limit.}  At first order, $\prod_i(1-D\lambda_i)^{-1}\approx1+Du$
and $u_\star\to\infty$, so the $u$-integral is analytic via
$\mathbb{E}[u\,\mathbf{1}[u>a]]=\sigma_\delta\,\phi(a/\sigma_\delta)$.  With
$a=-3\mu_3$,
\begin{eqnarray}
  \fMC &\approx& \FVC\bigl(1 + D\,\langle I_1\rangle_{\rm C}\bigr), \\[2mm]
  \langle I_1\rangle_{\rm C}
    &=& \frac{\sigma_\delta\,\mathbb{E}_{v,t}[\phi(3\mu_3/\sigma_\delta)]}{\FVC},
  \label{eq:I1C}
\end{eqnarray}
the inverse-Mills ratio averaged over the Doroshkevich $(v,t)$ measure.  The
$\mathbb{Z}_2$ antisymmetry $\langle I_1\rangle_{\rm C}+\langle I_1\rangle_{\rm V}=0$
follows from Eq.~\eqref{eq:t_reversal}.  The values quoted in
Eq.~\eqref{eq:I1means} are obtained by Monte Carlo evaluation of
Eq.~\eqref{eq:I1C}.


\bibliography{cosmic_web}

@article{Shannon1948,
    author = "Shannon, Claude E.",
    title = "{A mathematical theory of communication}",
    doi = "10.1002/j.1538-7305.1948.tb01338.x",
    journal = "Bell Syst. Tech. J.",
    volume = "27",
    number = "3",
    pages = "379--423",
    year = "1948"
}

@book{Cover2006,
  author    = {Cover, Thomas M. and Thomas, Joy A.},
  title     = {Elements of Information Theory},
  edition   = {2nd},
  publisher = {Wiley-Interscience},
  address   = {Hoboken, NJ},
  year      = {2006}
}

@article{Zeldovich1970,
  author    = {Zel'dovich, Ya.~B.},
  title     = {Gravitational instability: An approximate theory for large density perturbations},
  journal   = {Astronomy \& Astrophysics},
  year      = {1970},
  volume    = {5},
  pages     = {84--89},
  doi    = {https://ui.adsabs.harvard.edu/abs/1970A&A.....5...84Z},
}

@article{Bond1996,
    author = "Bond, J. Richard and Kofman, Lev and Pogosyan, Dmitri",
    title = "{How filaments are woven into the cosmic web}",
    eprint = "astro-ph/9512141",
    archivePrefix = "arXiv",
    reportNumber = "CITA-95-16",
    doi = "10.1038/380603a0",
    journal = "Nature",
    volume = "380",
    pages = "603--606",
    year = "1996"
}

@article{Hahn2007,
    author = "Hahn, Oliver and Porciani, Cristiano and Carollo, C. Marcella and Dekel, Avishai",
    title = "{Properties of Dark Matter Haloes in Clusters, Filaments, Sheets and Voids}",
    eprint = "astro-ph/0610280",
    archivePrefix = "arXiv",
    doi = "10.1111/j.1365-2966.2006.11318.x",
    journal = "Mon. Not. Roy. Astron. Soc.",
    volume = "375",
    pages = "489--499",
    year = "2007"
}

@article{Cautun2013,
    author = "Cautun, Marius and van de Weygaert, Rien and Jones, Bernard J. T.",
    title = "{NEXUS: Tracing the Cosmic Web Connection}",
    eprint = "1209.2043",
    archivePrefix = "arXiv",
    primaryClass = "astro-ph.CO",
    doi = "10.1093/mnras/sts416",
    journal = "Mon. Not. Roy. Astron. Soc.",
    volume = "429",
    pages = "1286",
    year = "2013"
}

@article{Cautun2014,
    author = "Cautun, Marius and van de Weygaert, Rien and Jones, Bernard J. T. and Frenk, Carlos S.",
    title = "{Evolution of the cosmic web}",
    eprint = "1401.7866",
    archivePrefix = "arXiv",
    primaryClass = "astro-ph.CO",
    doi = "10.1093/mnras/stu768",
    journal = "Mon. Not. Roy. Astron. Soc.",
    volume = "441",
    number = "4",
    pages = "2923--2973",
    year = "2014"
}

@article{Springel2005,
    author = "Springel, Volker and others",
    title = "{Simulating the joint evolution of quasars, galaxies and their large-scale distribution}",
    eprint = "astro-ph/0504097",
    archivePrefix = "arXiv",
    doi = "10.1038/nature03597",
    journal = "Nature",
    volume = "435",
    pages = "629--636",
    year = "2005"
}

@article{Springel2006,
    author = "Springel, Volker and Frenk, Carlos S. and White, Simon D. M.",
    title = "{The large-scale structure of the Universe}",
    eprint = "astro-ph/0604561",
    archivePrefix = "arXiv",
    doi = "10.1038/nature04805",
    journal = "Nature",
    volume = "440",
    pages = "1137",
    year = "2006"
}

@article{Springel2018,
    author = "Springel, Volker and others",
    title = "{First results from the IllustrisTNG simulations: matter and galaxy clustering}",
    eprint = "1707.03397",
    archivePrefix = "arXiv",
    primaryClass = "astro-ph.GA",
    doi = "10.1093/mnras/stx3304",
    journal = "Mon. Not. Roy. Astron. Soc.",
    volume = "475",
    number = "1",
    pages = "676--698",
    year = "2018"
}

@article{Alonso2013,
    author = "Alonso, D. and Bueno belloso, A. and S{\'a}nchez, F. J. and Garc{\'\i}a-Bellido, J. and S{\'a}nchez, E.",
    title = "{Measuring the transition to homogeneity with photometric redshift surveys}",
    eprint = "1312.0861",
    archivePrefix = "arXiv",
    primaryClass = "astro-ph.CO",
    doi = "10.1093/mnras/stu255",
    journal = "Mon. Not. Roy. Astron. Soc.",
    volume = "440",
    number = "1",
    pages = "10--23",
    year = "2014"
}

@article{Doroshkevich1970,
  author    = {Doroshkevich, A.~G.},
  title     = {Spatial structure of perturbations and origin of galactic rotation
               in fluctuation theory},
  journal   = {Astrophysics},
  year      = {1970},
  volume    = {6},
  pages     = {320--330},
  doi       = {10.1007/BF01001625}
}

@article{Doroshkevich1978,
  author    = {Doroshkevich, A.~G. and Shandarin, S.~F.},
  title     = {Spatial structure of perturbations and origin of galactic rotation
               in fluctuation theory},
  journal   = {Soviet Astronomy},
  year      = {1978},
  volume    = {22},
  pages     = {653--660},
  doi       = {https://ui.adsabs.harvard.edu/scan/manifest/1978SvA....22..653D}
}

@article{Bardeen1986,
  author    = {Bardeen, J.~M. and Bond, J.~R. and Kaiser, N. and Szalay, A.~S.},
  title     = {The statistics of peaks of {Gaussian} random fields},
  journal   = {Astrophysical Journal},
  year      = {1986},
  volume    = {304},
  pages     = {15--61},
  doi       = {10.1086/164143}
}

@article{Lee2000,
    author = "Lee, Jounghun and Pen, Ue-Li",
    title = "{Cosmic shear from galaxy spins}",
    eprint = "astro-ph/9911328",
    archivePrefix = "arXiv",
    reportNumber = "ASIAA-99-60",
    doi = "10.1086/312556",
    journal = "Astrophysical Journal Letters",
    volume = "532",
    pages = "L5--L8",
    year = "2000"
}

@article{Pogosyan1998,
  author    = {Pogosyan, D. and Bond, J.~R. and Kofman, L. and Primack, J.},
  title     = {Cosmic web: Origin and observables},
  booktitle = {Large Scale Structure: Tracks and Traces},
  year      = {1998},
  editor    = {M\"uller, V. and Kates, R. and Gottl\"ober, S.},
  publisher = {World Scientific},
  journal   = {World Scientific},
  pages     = {61--79},
  eprint    = {astro-ph/9810072},
  archivePrefix = {arXiv}
}

@article{Porqueres2021,
    author = "Porqueres, Natalia and Heavens, Alan and Mortlock, Daniel and Lavaux, Guilhem",
    title = "{Bayesian forward modelling of cosmic shear data}",
    eprint = "2011.07722",
    archivePrefix = "arXiv",
    primaryClass = "astro-ph.CO",
    doi = "10.1093/mnras/stab204",
    journal = "Mon. Not. Roy. Astron. Soc.",
    volume = "502",
    number = "2",
    pages = "3035--3044",
    year = "2021"
}

@article{Coles1991,
  author    = {Coles, Peter and Jones, Bernard},
  title     = {A lognormal model for the cosmological mass distribution},
  journal   = {Mon. Not. Roy. Astron. Soc.},
  year      = {1991},
  volume    = {248},
  pages     = {1--13},
  doi       = {10.1093/mnras/248.1.1},
    adsurl  = {https://ui.adsabs.harvard.edu/abs/1991MNRAS.248....1C}
}

@article{Bernardeau1994,
    author = "Bernardeau, Francis",
    title = "{Skewness and Kurtosis in large scale cosmic fields}",
    eprint = "astro-ph/9312026",
    archivePrefix = "arXiv",
    reportNumber = "CITA-93-44",
    doi = "10.1086/174620",
    journal = "Astrophysical Journal",
    volume = "433",
    pages = "1--18",
    year = "1994"
}

@article{Bernardeau2002,
    author = "Bernardeau, F. and Colombi, S. and Gazta{\~n}aga, E. and Scoccimarro, R.",
    title = "{Large scale structure of the universe and cosmological perturbation theory}",
    eprint = "astro-ph/0112551",
    archivePrefix = "arXiv",
    reportNumber = "SACLAY-T01-142",
    doi = "10.1016/S0370-1573(02)00135-7",
    journal = "Physics Reports",
    volume = "367",
    pages = "1--248",
    year = "2002"
}

@book{Mandelbrot1982,
  author    = {Mandelbrot, Beno{\^i}t B.},
  title     = {The Fractal Geometry of Nature},
  publisher = {W.~H.\ Freeman},
  address   = {New York},
  year      = {1982}
}

@book{Martinez2002,
  author    = {Mart\'{\i}nez, V.~J. and Saar, E.},
  title     = {Statistics of the Galaxy Distribution},
  publisher = {Chapman \& Hall/CRC},
  address   = {Boca Raton, FL},
  isbn      = {1584880848},
  year      = {2002}
}

@article{Hentschel1983,
  author    = {Hentschel, H.~G.~E. and Procaccia, Itamar},
  title     = {The infinite number of generalized dimensions of fractals and
               strange attractors},
  journal   = {Physica D},
  year      = {1983},
  volume    = {8},
  pages     = {435--444},
  doi       = {10.1016/0167-2789(83)90235-X}
}

@article{Halsey1986,
    author = "Halsey, Thomas C. and Jensen, Mogens H. and Kadanoff, Leo P. and Procaccia, Itamar and Shraiman, Boris I.",
    title = "{Fractal measures and their singularities: The characterization of strange sets}",
    doi = "10.1103/PhysRevA.33.1141",
    journal = "Physical Review A",
    volume = "33",
    pages = "1141--1151",
    year = "1986",
    note = "[Erratum: Phys.Rev.A 34, 1601 (1986)]"
}

@book{Mehta2004,
  author    = {Mehta, Madan Lal},
  title     = {Random Matrices},
  edition   = {3rd},
  publisher = {Elsevier/Academic Press},
  address   = {Amsterdam},
  year      = {2004}
}

@article{Leclercq2015,
    author = "Leclercq, Florent and Jasche, Jens and Wandelt, Benjamin",
    title = "{Bayesian analysis of the dynamic cosmic web in the SDSS galaxy survey}",
    eprint = "1502.02690",
    archivePrefix = "arXiv",
    primaryClass = "astro-ph.CO",
    doi = "10.1088/1475-7516/2015/06/015",
    journal = "JCAP",
    volume = "06",
    pages = "015",
    year = "2015"
}

@article{Jasche2013,
    author = "Jasche, Jens and Wandelt, Benjamin D.",
    title = "{Bayesian physical reconstruction of initial conditions from large scale structure surveys}",
    eprint = "1203.3639",
    archivePrefix = "arXiv",
    primaryClass = "astro-ph.CO",
    doi = "10.1093/mnras/stt449",
    journal = "Mon. Not. Roy. Astron. Soc.",
    volume = "432",
    pages = "894",
    year = "2013"
}

@book{Peebles1980,
  author    = {Peebles, P.~J.~E.},
  title     = {The Large-Scale Structure of the Universe},
  publisher = {Princeton University Press},
  address   = {Princeton, NJ},
  year      = {1980}
}

@article{DESI2016,
  author       = {{DESI Collaboration} and Aghamousa, A. and others},
  title        = {The {DESI} Experiment Part {I}: Science, Targeting, and Survey Design},
  journal      = {arXiv e-prints},
  year         = {2016},
  eprint       = {1611.00036},
  archivePrefix= {arXiv},
  primaryClass = {astro-ph.IM}
}

@article{Laureijs2011,
  author       = {Laureijs, R. and Amiaux, J. and Arduini, S. and others},
  title        = {{Euclid} Definition Study Report},
  journal      = {arXiv e-prints},
  year         = {2011},
  eprint       = {1110.3193},
  archivePrefix= {arXiv},
  primaryClass = {astro-ph.CO}
}

@article{Ivezic2019,
    author = "Ivezi{\'c}, {\v Z}. and Kahn, S.~M. and Tyson, J.~A. and others",
    collaboration = "LSST",
    title = "{LSST: from Science Drivers to Reference Design and Anticipated Data Products}",
    eprint = "0805.2366",
    archivePrefix = "arXiv",
    primaryClass = "astro-ph",
    reportNumber = "SLAC-PUB-16076",
    doi = "10.3847/1538-4357/ab042c",
    journal = "Astrophysical Journal",
    volume = "873",
    number = "2",
    pages = "111",
    year = "2019"
}

@article{Linder2005,
    author = "Linder, Eric V.",
    title = "{Cosmic growth history and expansion history}",
    eprint = "astro-ph/0507263",
    archivePrefix = "arXiv",
    doi = "10.1103/PhysRevD.72.043529",
    journal = "Phys. Rev. D",
    volume = "72",
    pages = "043529",
    year = "2005"
}

@article{Hosoya2004,
    author = "Hosoya, Akio and Buchert, Thomas and Morita, Masaaki",
    title = "{Information entropy in cosmology}",
    eprint = "gr-qc/0402076",
    archivePrefix = "arXiv",
    doi = "10.1103/PhysRevLett.92.141302",
    journal = "Phys. Rev. Lett.",
    volume = "92",
    pages = "141302",
    year = "2004"
}

@article{Pandey2013,
    author = "Pandey, Biswajit",
    title = "{A method for testing the cosmic homogeneity with Shannon entropy}",
    eprint = "1301.4961",
    archivePrefix = "arXiv",
    primaryClass = "astro-ph.CO",
    doi = "10.1093/mnras/stt134",
    journal = "Mon. Not. Roy. Astron. Soc.",
    volume = "430",
    pages = "3376",
    year = "2013"
}

@article{Pandey2015,
    author = "Pandey, Biswajit and Sarkar, Suman",
    title = "{Testing homogeneity in the Sloan Digital Sky Survey Data Release Twelve with Shannon entropy}",
    eprint = "1507.03124",
    archivePrefix = "arXiv",
    primaryClass = "astro-ph.CO",
    doi = "10.1093/mnras/stv2166",
    journal = "Mon. Not. Roy. Astron. Soc.",
    volume = "454",
    number = "3",
    pages = "2647--2656",
    year = "2015"
}

@article{Leclercq2016,
    author = "Leclercq, Florent and Lavaux, Guilhem and Jasche, Jens and Wandelt, Benjamin",
    title = "{Comparing cosmic web classifiers using information theory}",
    eprint = "1606.06758",
    archivePrefix = "arXiv",
    primaryClass = "astro-ph.CO",
    doi = "10.1088/1475-7516/2016/08/027",
    journal = "Journal of Cosmology and Astroparticle Physics",
    volume = "08",
    pages = "027",
    year = "2016"
}

@article{Falck:2012ai,
    author = "Falck, Bridget L. and Neyrinck, Mark C. and Szalay, Alexander S.",
    title = "{ORIGAMI: Delineating Halos using Phase-Space Folds}",
    eprint = "1201.2353",
    archivePrefix = "arXiv",
    primaryClass = "astro-ph.CO",
    doi = "10.1088/0004-637X/754/2/126",
    journal = "Astrophys. J.",
    volume = "754",
    pages = "126",
    year = "2012"
}

@article{Vazza2017,
    author = "Vazza, Franco",
    title = "{On the complexity and the information content of cosmic structures}",
    eprint = "1611.09348",
    archivePrefix = "arXiv",
    primaryClass = "astro-ph.CO",
    doi = "10.1093/mnras/stw3089",
    journal = "Mon. Not. Roy. Astron. Soc.",
    volume = "465",
    number = "4",
    pages = "4942--4955",
    year = "2017"
}

@article{Vazza2020,
    author = "Vazza, F.",
    title = "{How complex is the cosmic web?}",
    eprint = "1911.11029",
    archivePrefix = "arXiv",
    primaryClass = "astro-ph.CO",
    doi = "10.1093/mnras/stz3317",
    journal = "Mon. Not. Roy. Astron. Soc.",
    volume = "491",
    number = "4",
    pages = "5447--5463",
    year = "2020"
}

@article{Gaite2019,
  author    = {Gaite, J.},
  title     = {The fractal geometry of the cosmic web and its formation},
  journal   = {Advances in Astronomy},
  year      = {2019},
  volume    = {1},
  pages     = {1--25},
  doi       = {10.1155/2019/6587138}
}

@article{Bonnaire2022,
    author = "Bonnaire, Tony and Aghanim, Nabila and Kuruvilla, Joseph and Decelle, Aur{\'e}lien",
    title = "{Cosmology with cosmic web environments - I. Real-space power spectra}",
    eprint = "2112.03926",
    archivePrefix = "arXiv",
    primaryClass = "astro-ph.CO",
    doi = "10.1051/0004-6361/202142852",
    journal = "Astronomy \& Astrophysics",
    volume = "661",
    pages = "A146",
    year = "2022"
}

@article{Coutinho2016,
  author    = {Coutinho, B.~C. and Hong, S. and Albrecht, K. and Dey, A. and
               Barab{\'a}si, A.-L. and Torrey, P. and Vogelsberger, M. and
               Hernquist, L.},
  title     = {The network behind the cosmic web},
  journal   = {arXiv e-prints},
  year      = {2016},
  eprint    = {1604.03236},
  archivePrefix = {arXiv},
  primaryClass  = {astro-ph.CO}
}

@article{Kitaura2026,
    author = "Kitaura, Francisco-Shu",
    title = "{Emergence of complex web structures}",
    eprint = "2604.11481",
    archivePrefix = "arXiv",
    primaryClass = "astro-ph.CO",
    doi = "10.1088/1475-7516/2026/07/051",
    journal = "JCAP",
    volume = "07",
    pages = "051",
    year = "2026"
}

@article{Shandarin1983,
    author = "Shandarin, S. F. and Doroshkevich, A. G. and Zeldovich, Ya. B.",
    title = "{The large-scale structure of the universe}",
    doi = "10.1070/PU1983v026n01ABEH004305",
    journal = "Sov. Phys. Usp.",
    volume = "26",
    pages = "46--76",
    year = "1983"
}

@article{Maldacena2003,
    author = "Maldacena, Juan Martin",
    title = "{Non-Gaussian features of primordial fluctuations in single field inflationary models}",
    eprint = "astro-ph/0210603",
    archivePrefix = "arXiv",
    doi = "10.1088/1126-6708/2003/05/013",
    journal = "JHEP",
    volume = "05",
    pages = "013",
    year = "2003"
}

@article{Verde2000,
    author = "Verde, Licia and Wang, Li-Min and Heavens, Alan and Kamionkowski, Marc",
    title = "{Large scale structure, the cosmic microwave background, and primordial non-gaussianity}",
    eprint = "astro-ph/9906301",
    archivePrefix = "arXiv",
    doi = "10.1046/j.1365-8711.2000.03191.x",
    journal = "Mon. Not. Roy. Astron. Soc.",
    volume = "313",
    pages = "L141--L147",
    year = "2000"
}

@article{Dalal2008,
    author = "Dalal, Neal and Dore, Olivier and Huterer, Dragan and Shirokov, Alexander",
    title = "{The imprints of primordial non-gaussianities on large-scale structure: scale dependent bias and abundance of virialized objects}",
    eprint = "0710.4560",
    archivePrefix = "arXiv",
    primaryClass = "astro-ph",
    doi = "10.1103/PhysRevD.77.123514",
    journal = "Phys. Rev. D",
    volume = "77",
    pages = "123514",
    year = "2008"
}

@article{Matarrese2008,
    author = "Matarrese, Sabino and Verde, Licia",
    title = "{The effect of primordial non-Gaussianity on halo bias}",
    eprint = "0801.4826",
    archivePrefix = "arXiv",
    primaryClass = "astro-ph",
    doi = "10.1086/587840",
    journal = "Astrophys. J. Lett.",
    volume = "677",
    pages = "L77--L80",
    year = "2008"
}

@article{Desjacques2018,
    author = "Desjacques, Vincent and Jeong, Donghui and Schmidt, Fabian",
    title = "{Large-Scale Galaxy Bias}",
    eprint = "1611.09787",
    archivePrefix = "arXiv",
    primaryClass = "astro-ph.CO",
    doi = "10.1016/j.physrep.2017.12.002",
    journal = "Phys. Rept.",
    volume = "733",
    pages = "1--193",
    year = "2018"
}

@article{Libeskind2018,
    author = "Libeskind, Noam I. and others",
    title = "{Tracing the cosmic web}",
    eprint = "1705.03021",
    archivePrefix = "arXiv",
    primaryClass = "astro-ph.CO",
    doi = "10.1093/mnras/stx1976",
    journal = "Mon. Not. Roy. Astron. Soc.",
    volume = "473",
    pages = "1195--1217",
    year = "2018"
}

@article{ForeRomero2009,
    author = "Forero-Romero, J. E. and Hoffman, Y. and Gottloeber, S. and Klypin, A. and Yepes, G.",
    title = "{A dynamical classification of the cosmic web}",
    eprint = "0907.1321",
    archivePrefix = "arXiv",
    primaryClass = "astro-ph.CO",
    doi = "10.1111/j.1365-2966.2009.15191.x",
    journal = "Mon. Not. Roy. Astron. Soc.",
    volume = "396",
    pages = "1815--1824",
    year = "2009"
}

@article{Hoffman2012,
    author = "Hoffman, Yehuda and Metuki, Oren and Yepes, Gustavo and Gottloeber, Stefan and Forero-Romero, Jaime E. and Libeskind, Noam I. and Knebe, Alexander",
    title = "{A kinematic classification of the cosmic web}",
    eprint = "1201.3367",
    archivePrefix = "arXiv",
    primaryClass = "astro-ph.CO",
    doi = "10.1111/j.1365-2966.2012.21553.x",
    journal = "Mon. Not. Roy. Astron. Soc.",
    volume = "425",
    pages = "2049--2057",
    year = "2012"
}

@article{Sousbie2011,
    author = "Sousbie, Thierry",
    title = "{The persistent cosmic web and its filamentary structure I: Theory and implementation}",
    eprint = "1009.4015",
    archivePrefix = "arXiv",
    primaryClass = "astro-ph.CO",
    doi = "10.1111/j.1365-2966.2011.18394.x",
    journal = "Mon. Not. Roy. Astron. Soc.",
    volume = "414",
    pages = "350--383",
    year = "2011"
}

@article{AragonCalvo2010,
    author = "Aragon-Calvo, Miguel A. and van de Weygaert, Rien and Jones, Bernard J.T.",
    title = "{Multiscale phenomenology of the cosmic web}",
    eprint = "1007.0742",
    archivePrefix = "arXiv",
    primaryClass = "astro-ph.CO",
    doi = "10.1111/j.1365-2966.2010.16945.x",
    journal = "Mon. Not. Roy. Astron. Soc.",
    volume = "408",
    pages = "2163--2187",
    year = "2010"
}

@article{Novikov2006,
    author = "Novikov, Dmitri and Colombi, St{\'e}phane and Dor{\'e}, Olivier",
    title = "{Skeleton as a probe of the cosmic web: the two-dimensional case}",
    eprint = "astro-ph/0307003",
    archivePrefix = "arXiv",
    primaryClass = "astro-ph",
    doi = "10.1111/j.1365-2966.2005.09925.x",
    journal = "Mon. Not. Roy. Astron. Soc.",
    volume = "366",
    pages = "1201--1216",
    year = "2006"
}

@article{Sousbie2008,
    author = "Sousbie, Thierry and Pichon, Christophe and Colombi, St{\'e}phane and Novikov, Dmitri and Pogosyan, Dmitri",
    title = "{The 3D skeleton: tracing the filamentary structure of the Universe}",
    eprint = "0707.3123",
    archivePrefix = "arXiv",
    primaryClass = "astro-ph",
    doi = "10.1111/j.1365-2966.2007.12685.x",
    journal = "Mon. Not. Roy. Astron. Soc.",
    volume = "383",
    pages = "1655--1670",
    year = "2008"
}

@article{Tempel2015,
    author = "Tempel, Elmo and Stoica, Radu S. and Martinez, Vicent J. and Liivamagi, Lauri J. and Castellan, Gwenael and Saar, Enn",
    title = "{Detecting filamentary pattern in the cosmic web: a catalogue of filaments for the SDSS}",
    eprint = "1406.5549",
    archivePrefix = "arXiv",
    primaryClass = "astro-ph.CO",
    doi = "10.1093/mnras/stu2368",
    journal = "Mon. Not. Roy. Astron. Soc.",
    volume = "438",
    pages = "3465--3482",
    year = "2015"
}

@article{Contarini2023,
    author = "Contarini, Sofia and Pisani, Alice and Hamaus, Nico and Marulli, Federico and Moscardini, Lauro and Baldi, Marco",
    title = "{Cosmological constraints from the void size function}",
    eprint = "2212.03085",
    archivePrefix = "arXiv",
    primaryClass = "astro-ph.CO",
    doi = "10.1051/0004-6361/202244095",
    journal = "Astron. Astrophys.",
    volume = "667",
    pages = "A162",
    year = "2022"
}

@article{Cabass2022,
    author = "Cabass, Giovanni and Ivanov, Mikhail M. and Philcox, Oliver H.E. and Simonovic, Marko and Zaldarriaga, Matias",
    title = "{Constraining $\mu$-distortions and $r$ with the BOSS galaxy bispectrum}",
    eprint = "2204.01781",
    archivePrefix = "arXiv",
    primaryClass = "astro-ph.CO",
    doi = "10.1103/PhysRevD.106.043506",
    journal = "Phys. Rev. D",
    volume = "106",
    pages = "043506",
    year = "2022"
}

@article{DAmico2022,
    author = "D'Amico, Guido and Lewandowski, Matthew and Senatore, Leonardo and Zhang, Pierre",
    title = "{Limits on primordial non-Gaussianities from BOSS galaxy-clustering data}",
    eprint = "2201.11518",
    archivePrefix = "arXiv",
    primaryClass = "astro-ph.CO",
    doi = "10.1103/PhysRevD.111.063514",
    journal = "Phys. Rev. D",
    volume = "111",
    number = "6",
    pages = "063514",
    year = "2025"
}

@article{DESI2024PNG,
    author = "{DESI Collaboration} and Adame, A. G. and others",
    title = "{DESI 2024 VI: cosmological constraints from BAO measurements}",
    eprint = "2404.03002",
    archivePrefix = "arXiv",
    primaryClass = "astro-ph.CO",
    reportNumber = "FERMILAB-PUB-24-0154-PPD",
    doi = "10.1088/1475-7516/2025/02/021",
    journal = "JCAP",
    volume = "02",
    pages = "021",
    year = "2025"
}

@article{Xavier2016,
    author = "Xavier, Henrique S. and Abdalla, Filipe B. and Joachimi, Benjamin",
    title = "{Improving lognormal models for cosmological fields}",
    eprint = "1612.00065",
    archivePrefix = "arXiv",
    primaryClass = "astro-ph.CO",
    doi = "10.1093/mnras/stx1153",
    journal = "Mon. Not. Roy. Astron. Soc.",
    volume = "470",
    pages = "3131--3150",
    year = "2016"
}

@article{Euclid2024,
    author = "{Euclid Collaboration} and Mellier, Y. and others",
    title = "{Euclid. I. Overview of the Euclid mission}",
    eprint = "2405.13491",
    archivePrefix = "arXiv",
    primaryClass = "astro-ph.CO",
    doi = "10.1051/0004-6361/202450810",
    journal = "Astron. Astrophys.",
    volume = "697",
    pages = "A1",
    year = "2025"
}

@article{Carron2012,
    author = "Carron, Julien",
    title = "{On the assumption of Gaussianity for cosmological two-point statistics and parameter estimation}",
    eprint = "1204.4724",
    archivePrefix = "arXiv",
    primaryClass = "astro-ph.CO",
    doi = "10.1051/0004-6361/201116538",
    journal = "Astron. Astrophys.",
    volume = "551",
    pages = "A88",
    year = "2013"
}

@article{Kitaura:2020lkj,
    author = "Kitaura, Francisco-Shu and Balaguera-Antol{\'\i}nez, Andr{\'e}s and Sinigaglia, Francesco and Pellejero-Ib{\'a}{\~n}ez, Marcos",
    title = "{The cosmic web connection to the dark matter halo distribution through gravity}",
    eprint = "2005.11598",
    archivePrefix = "arXiv",
    primaryClass = "astro-ph.CO",
    doi = "10.1093/mnras/stac671",
    journal = "Mon. Not. Roy. Astron. Soc.",
    volume = "512",
    number = "2",
    pages = "2245--2265",
    year = "2022"
}

@article{Coloma-Nadal:2024ccw,
    author = "Coloma-Nadal, J. M. and Kitaura, F. -S. and Garc{\'\i}a-Farieta, J. E. and Sinigaglia, F. and Favole, G. and S{\'a}nchez, D. Forero",
    title = "{The hierarchical cosmic web and assembly bias}",
    eprint = "2403.19337",
    archivePrefix = "arXiv",
    primaryClass = "astro-ph.CO",
    doi = "10.1088/1475-7516/2024/07/083",
    journal = "JCAP",
    volume = "07",
    pages = "083",
    year = "2024"
}

@article{Kitaura:2026ike,
    author = "Kitaura, Francisco-Shu and Sinigaglia, Francesco",
    title = "{Spectral Hierarchy of the Cosmic Web}",
    eprint = "2603.15834",
    archivePrefix = "arXiv",
    primaryClass = "astro-ph.CO",
    month = "3",
    year = "2026"
}

@article{Kitaura:2010tr,
    author = "Kitaura, F. S.",
    title = "{Non-Gaussian gravitational clustering field statistics}",
    eprint = "1012.3168",
    archivePrefix = "arXiv",
    primaryClass = "astro-ph.CO",
    doi = "10.1111/j.1365-2966.2011.19680.x",
    journal = "Mon. Not. Roy. Astron. Soc.",
    volume = "420",
    pages = "2737",
    year = "2012"
}

@article{Pandey2019,
    author = "Pandey, Biswajit",
    title = "{Configuration entropy of the Cosmic Web: Can voids mimic the dark energy?}",
    eprint = "1901.08475",
    archivePrefix = "arXiv",
    primaryClass = "astro-ph.CO",
    doi = "10.1093/mnrasl/slz037",
    journal = "Mon. Not. Roy. Astron. Soc.",
    volume = "485",
    number = "1",
    pages = "L73--L77",
    year = "2019"
}

@article{Grossi2009,
  author  = {Grossi, M. and Verde, L. and Carbone, C. and Dolag, K. and
             Branchini, E. and Iannuzzi, F. and Matarrese, S. and Moscardini, L.},
  title   = {Large-scale non-Gaussian mass function and halo bias:
             tests on {N}-body simulations},
  journal = {Mon. Not. R. Astron. Soc.},
  volume  = {398},
  pages   = {321},
  year    = {2009},
  eprint  = {0902.2013},
  archivePrefix = {arXiv},
  primaryClass  = {astro-ph.CO},
  doi     = {10.1111/j.1365-2966.2009.15150.x}
}

@article{LoVerde2008,
  author  = {LoVerde, M. and Miller, A. and Shandarin, S. and Verde, L.},
  title   = {Effects of scale-dependent non-Gaussianity on cosmological structures},
  journal = {JCAP},
  volume  = {04},
  pages   = {014},
  year    = {2008},
  eprint  = {0711.4126},
  archivePrefix = {arXiv},
  primaryClass  = {astro-ph},
  doi     = {10.1088/1475-7516/2008/04/014}
}

\end{document}